\documentclass{article}

\usepackage[mathletters]{ucs}
\usepackage[utf8x]{inputenx}
\usepackage{amsmath,amsthm,amssymb,xspace,stmaryrd}
\usepackage{graphicx}
\usepackage[all,cmtip]{xy}
 \usepackage{comment}
 
\begin{document}

\newcommand{\drie}{\vartriangleleft}
\newcommand{\qup}{\underline{\mathbb{Q}}^{+}}
\newcommand{\qui}{\underline{\mathbb{Q}[i]}}
\newcommand{\uS}{\underline{\Sigma}}
\newcommand{\uO}{\underline{\Omega}}
\newcommand{\uN}{\underline{N}}
\newcommand{\uA}{\underline{A}}
\newcommand{\uB}{\underline{B}}
\newcommand{\topos}{[\mathcal{C}^{\text{op}},\mathbf{Set}]}
\newcommand{\idl}{\mathcal{I}(\mathcal{C})}
\newcommand{\fil}{\mathcal{F}(\mathcal{C})}
\newcommand{\Cd}{\mathcal{C}_{\downarrow}}
\newcommand{\Cu}{\mathcal{C}_{\uparrow}}
\newcommand{\SdA}{\Sigma_{A}^{\downarrow}}
\newcommand{\SdB}{\Sigma_{B}^{\downarrow}}
\newcommand{\SuA}{\Sigma_{A}^{\uparrow}}
\newcommand{\SuB}{\Sigma_{B}^{\uparrow}}

\newtheorem{dork}{Definition}[section]
\newtheorem{tut}[dork]{Theorem}
\newtheorem{poe}[dork]{Proposition}
\newtheorem{lem}[dork]{Lemma}
\newtheorem{cor}[dork]{Corollary}
\newtheorem{con}[dork]{Conjecture}
\newtheorem{rem}[dork]{Remark}
\newtheorem{exa}[dork]{Example}

\title{Independence Conditions for Nets of Local Algebras as Sheaf Conditions}
 \author{Sander Wolters\footnote{Radboud Universiteit Nijmegen, Institute for Mathematics, Astrophysics, and Particle Physics, the Netherlands. s.wolters@math.ru.nl. Supported by
N.W.O. through project 613.000.811.} \and Hans Halvorson\footnote{Department of Philosophy, Princeton University, Princeton NJ, US. hhalvors@princeton.edu}}
\maketitle
\begin{center}
\textit{Dedicated to Joost Nuiten}
\end{center}

\begin{abstract}
We apply constructions from topos-theoretic approaches to quantum theory to algebraic quantum field theory. Thus a net of operator algebras is reformulated as a functor that maps regions of spacetime into a category of ringed topoi. We ask whether this functor is a sheaf, a question which is related to the net satisfying certain kinematical independence conditions. In addition, we consider a C*-algebraic version of Nuiten's recent sheaf condition, and demonstrate how it relates to C*-independence of the underlying net of operator algebras.
\end{abstract}

\section{Introduction and Motivation} \label{sec: intro}

In this section we briefly describe and motivate the material in the later sections. The reader is assumed to have some familiarity with topos theory and C*-algebras.

\subsection{Topos Models for Quantum Physics}

To be more precise, by topos-theoretic approaches to quantum physics we mean work inspired by the ideas of Butterfield and Isham, \cite{butish1,butish2,buhais,butish3}, who show that elements of topos theory appear naturally in the foundations of quantum mechanics. This work, which was further developed by D\"oring and Isham~\cite{di1,di2,di3,di4} concentrates on a particular topos, namely a category of presheaves. Typically, this topos is defined using von Neumann algebras, but we shall make use of the larger class of unital C*-algebras\footnote{The additional structure of von Neumann algebras is important when considering daseinisation~\cite{di2}. As we do not consider daseinisation here, C*-algebras give greater generality.}. We concentrate on a quantum system described by a unital C*-algebra $A$. Let $\mathcal{C_{A}}$ (or simply $\mathcal{C}$) denote the set of unital commutative C*-subalgebras of $A$. The elements of $\mathcal{C}$ are called \textbf{contexts}, and we consider $\mathcal{C}$ as a poset, with the order given by inclusion. The topos used by Butterfield, D\"oring and Isham is the category $[\mathcal{C}^{op},\mathbf{Set}]$ of presheaves on $\mathcal{C}$.

The work of Isham and D\"oring shows great ambition in using topos theory for physics. One central idea~\cite{di1,ish} is that any theory of physics, in its mathematical formulation should have certain structures, which make the mathematical framework resemble that of classical physics in some sense. This, in turn, should help in giving some, hopefully not-naive, realist account of the theory. In addition to restricting the mathematical shape of physical theories, considerable freedom is added by allowing these structures to be interpreted in topoi other than the category of sets. In this direction, the motivating examples are the presheaf models $[\mathcal{C}^{op}_{A},\mathbf{Set}]$, for varying $A$. 

A different topos model for quantum theories is given by the copresheaf model introduced by Heunen, Landsman and Spitters~\cite{hls}. This model, which will be used extensively in later sections, uses the topos $[\mathcal{C}_{A},\mathbf{Set}]$ of copresheaves\footnote{i.e. the category covariant functors $\mathcal{C}_{A}\to\mathbf{Set}$ and their natural transformations}. An interesting feature of the covariant model of Heunen et al. is that it relies on the internal language of the copresheaf topos. A topos, which is a highly structured category, can be thought of as a mathematical universe of discourse. It is a place where we can do mathematics, in a way that strongly resembles set-based mathematics. It is from this (internal) perspective that the copresheaf topos model of the quantum system resembles classical physics. The presheaf model by Isham et.al. turns out to be closely related to the copresheaf model~\cite{wollie}, and hence also allows such an internal perspective which resembles classical physics~\cite{wollie2}.

We concentrate mostly on the object of $[\mathcal{C}_{A},\mathbf{Set}]$, defined by the functor
\begin{equation*}
\underline{A}:\mathcal{C}\to\mathbf{Set}\ \ \ \underline{A}(C)=C,
\end{equation*}
where, if $D\subseteq C$, the corresponding arrow in $\mathcal{C}$ is mapped to the inclusion $D\hookrightarrow C$. A key observation is that from the internal perspective of the topos $[\mathcal{C}_{A},\mathbf{Set}]$ this object is a unital commutative C*-algebra~\cite{hls}. 

\subsection{Nets of Operator Algebras as Functors}

Let $\mathcal{V}(X)$ denote a poset of causally complete opens of a spacetime manifold $X$, partially ordered by inclusion. We assume that $\mathcal{V}(X)$, when ordered by inclusion, is a lattice\footnote{That is, any pair $O_{1},O_{2}\in\mathcal{V}(X)$ has a greatest lower bound $O_{1}\wedge O_{2}$, and a least upper bound $O_{1}\vee O_{2}$.}. Assume that we have a mapping $O\mapsto A(O)$, associating to each $O\in\mathcal{V}(X)$ a unital C*-algebra $A(O)$. Further assume that the map $O\mapsto A(O)$ is isotonic, in the sense that if $O_{1}\subseteq O_{2}$, then $A(O_{1})\subseteq A(O_{2})$. Such a mapping $O\mapsto A(O)$ is called a \textit{net of operator algebras}.

For each $O\in\mathcal{V}(X)$, the topos approach of~\cite{hls} provides us with a pair $(\mathcal{T}_{O},\underline{A(O)})$ consisting of a topos $\mathcal{T}_{O}=[\mathcal{C}_{A(O)},\mathbf{Set}]$ and a unital commutative C*-algebra $\underline{A(O)}$ in this topos.

As shown in detail in Section~\ref{sec: sheaf}, if $O_{1}\subseteq O_{2}$, we can associate to this inclusion a pair
\begin{equation} \label{equ: arroe}
(I,\underline{i}):(\mathcal{T}_{O_{2}},\underline{A(O_{2})})\to(\mathcal{T}_{O_{1}},\underline{A(O_{1})}),
\end{equation}
where $I:\mathcal{T}_{O_{2}}\to\mathcal{T}_{O_{1}}$ is a geometric morphism\footnote{A geometric morphism is a morphism of topoi, defined as an adjunction $I^{\ast}\dashv I_{\ast}$, where the left adjoint $I^{\ast}:\mathcal{T}_{O_{2}}\to\mathcal{T}_{O_{1}}$ is left-exact.}, and $\underline{i}: I^{\ast}\underline{A(O_{1})}\to\underline{A(O_{2})}$ is a $\ast$-homomorphism , internal to the topos $\mathcal{T}_{O_{2}}$.

Associating to each $O$ the pair $(\mathcal{T}_{O},\underline{A(O)})$, and to each inclusion $O_{1}\subseteq O_{2}$ a pair of arrows (\ref{equ: arroe}), a net $O\mapsto A(O)$ defines a contravariant functor from $\mathcal{V}(X)$ to the category $\mathbf{ucCTopos}$ of topoi with internal unital commutative C*-algebras. However, for technical reasons a different target category $\mathbf{RingSp}$ is used. This category differs from $\mathbf{ucCTopos}$ in two ways. Instead of using all (Grothendieck) topoi, only topoi of the form $Sh(X)$, the topos of sheaves on a topological space $X$, will be considered. In addition, instead of using internal commutative C*-algebras, the more general internal commutative rings are used. In Section~\ref{sec: pullback} we will drop this last condition and consider a fully C*-algebraic version.

Viewing the net $O\mapsto A(O)$ as a functor $\mathcal{V}(X)^{op}\to\mathbf{RingSp}$, and noting that $\mathbf{RingSp}$ is complete as a category, we can ask if this functor is a sheaf. Of course, we need to specify something like a covering relation before we can ask this question. Let $O_{1}$ and $O_{2}$ be two spacelike separated regions of spacetime. Then we say that $O:=O_{1}\vee O_{2}$ is covered by $O_{1}$ and $O_{2}$.

The investigation of this sheaf condition was first performed by Nuiten~\cite{nuiten}. This work is impressive, especially if you take into account that the research was done for Nuiten's bachelor thesis, under the supervision of Dr. Urs Schreiber. As it turns out, for physically reasonable nets the functor  $\mathcal{V}(X)^{op}\to\mathbf{RingSp}$ is not a sheaf. But, it does come close to a sheaf. Technically, what is meant by this is that under a mild condition the descent morphism is a local geometric surjection. This condition is called strong locality. Strong locality implies microcausality, and is implied by C*-independence in the product sense (see Definition~\ref{def: list}, taking $A=A(O_{1})$ and $B=A(O_{2})$, letting $O_{1}$ and $O_{2}$ be spacelike separated).

For the C*-algebraic version $\mathcal{V}(X)^{op}\to\mathbf{ucCSp}$, defined in Section~\ref{sec: pullback}, the sheaf condition is shown to be equivalent to C*-independence of the net together with the condition
\begin{equation*}
\forall C\in\mathcal{C}_{A(O_{1}\vee O_{2})}\ \ (C\cap A(O_{1}))\vee(C\cap A(O_{2}))=C,
\end{equation*}
for all pairs $(O_{1},O_{2})$ of spacelike separated regions.

\subsection{Motivation}

Why do we consider these constructions to be of interest? For a moment, suppose that we are sceptical about the specific topos models for quantum physics constructed so far~\cite{butish1,di,hls}. Even so, the discussion in~\cite{butish1} may suggest that using the language of (pre)sheaves over posets of contexts is a natural step in studying the foundations of quantum theory. Indeed, the work~\cite{ab,amsb} studies non-locality and contextuality using (pre)sheaves, without invoking any topos theory. Furthermore, studying the relation between the poset $\mathcal{C}_{A}$ and the algebra $A$ may be of interest in itself~\cite{hamhalter2,hado}. 

The research described below fits nicely within such programs. As an example, consider the following simple lemma by Nuiten~\cite{nuiten}. Let $(A,B)$ be a pair of C*-algebras associated to spacelike separated regions by some net. This net is said to satisfy microcausality iff such algebras commute, i.e. $[A,B]=\{0\}$. This condition of microcausality, then, is equivalent to the claim that the poset morphism 
\begin{equation*}
\mathcal{C}_{A\vee B}\to\mathcal{C}_{A}\times\mathcal{C}_{B},\ \ C\mapsto(C\cap A, C\cap B)
\end{equation*}
has a left adjoint. Hence, at the level of contexts, microcausality can be captured categorically!

Next, assume that we are curious about the ideas of Isham and D\"oring to the effect that physical theories should be formulated in possibly nontrivial (i.e. non-$\mathbf{Set}$) topoi. Apart from the copresheaf model $[\mathcal{C}_{A},\mathbf{Set}]$, insofar as these count, the only nontrivial example is the motivating $[\mathcal{C}_{A}^{op},\mathbf{Set}]$. The discussion given below may be of help in finding new topos models, as a central theme in the work is to encode independence conditions on nets of algebras as sheaf conditions. At this point, the reader who knows the basics of topos theory may object that the sheaves discussed below are not sheaves on a site~\cite[Chapter III]{mm}, but a generalisation thereof. Nevertheless, as suggested by Subsection~\ref{subsub: spectral}, there are relations between Nuiten's sheaves and sheaves on sites.

Finally, we can take the stance that we are interested in topos theory, but not so much in topos models to quantum theory. In this case, Section~\ref{sec: pullback} may be of interest when seen as a discussion of C*-algebras in topoi.\\

The text below is divided into four parts. In Section~\ref{sec: back} we briefly discuss those ideas from topos theory that are important for the other sections. In particular, we discuss C*-algebras in topoi. Section~\ref{sec: sheaf} discusses the sheaf condition of~\cite{nuiten}. One difference from the original treatment is that we do not assume the net of operator algebras to be additive. The sheaf condition leads to a new independence condition, called \textit{strong locality}. In Section~\ref{sec: locality} strong locality is compared to other locality conditions, such as microcausality and C*-independence. In relation to this, we describe these locality conditions at the level of commutative subalgebras $\mathcal{C}$. In Section~\ref{sec: pullback} we revisit the sheaf condition in a slightly different setting. In this setting rings are replaced by C*-algebras, which leads to certain technical complications. The sheaf condition is subsequently related to C*-independence of the net.

\section{Background} \label{sec: back}

In what follows we use rings (in particular C*-algebras) internal to a topos, geometric morphisms, and the action of geometric morphisms on internal C*-algebras. As remarked in the introduction, an important observation in the copresheaf topos model is that the object $\uA$ is a unital commutative C*-algebra in the topos $[\mathcal{C},\mathbf{Set}]$. The word `internal' refers to the internal language, or Mitchell-B\'enabou language, which can be associated to any topos~\cite[Section VI.5]{mm}. With respect to this language, an object $\uA$, together with arrows representing the algebraic structure such as addition $+:\uA\times\uA\to\uA$, and the involution $\ast:\uA\to\uA$, defines an internal C*-algebra if it satisfies all relevant axioms, such as
\begin{equation*}
\forall a,b,c\in\uA\ \left((a+b)+c=a+(b+c)\right),
\end{equation*}
expressing associativity of addition. In other words, $\uA$ is a model of the theory of C*-algebras in the topos $[\mathcal{C},\mathbf{Set}]$. This theory consists of a sort $A$, function symbols such as $+:A\times A\to A$ and $0:1\to A$, a predicate $N\subseteq A\times\mathbb{Q}^{+}$ for the norm, together with the axioms for a C*-algebra. In Appendix~\ref{sec: internal} we treat the axioms for an internal C*-algebra in some detail, with emphasis on the norm. For our purposes we do not need to know the details of this language. These details, and in particular, what it means for an axiom to be valid, can be found in various texts on topos theory. The text~\cite{bell} and Chapter VI of~\cite{mm} provide good introductions. The books~\cite{lasc,bor}, present many worked out examples of internal mathematics. The comprehensive~\cite{jh1} also explains the internal workings of topoi.

If $\mathcal{E}$ and $\mathcal{F}$ are topoi, then a geometric morphism $F:\mathcal{E}\to\mathcal{F}$ is an adjunction $F^{\ast}\dashv F_{\ast}$, where the left adjoint $F^{\ast}:\mathcal{F}\to\mathcal{E}$ is a left-exact functor, called the inverse image functor, and the right adjoint $F_{\ast}:\mathcal{E}\to\mathcal{F}$ is called the direct image functor.

\begin{exa}
To a topological $X$ we can associate the topos $Sh(X)$ of sheaves on that space~\cite[Chapter II]{mm}. If $X$ and $Y$ are Hausdorff spaces, geometric morphisms $F:Sh(X)\to Sh(Y)$ correspond bijectively to continuous maps $f:X\to Y$. If some geometric morphism $F$ corresponds to a map $f$, the corresponding direct image functor is given by
\begin{equation*}
\forall \uA\in Sh(X),\ \forall V\in\mathcal{O}Y,\ \  F_{\ast}(\uA)(V)=\uA(f^{-1}(V)).
\end{equation*}
If we describe sheaves as \'etale bundles, then the inverse image functor corresponds to the pullback of the \'etale bundles on $Y$ along the map $f:X\to Y$.
\end{exa}

\begin{exa}
Let $\mathcal{C}$ and $\mathcal{D}$ be small categories. A functor $\phi:\mathcal{C}\to\mathcal{D}$ induces a geometric morphism $\Phi:[\mathcal{C},\mathbf{Set}]\to[\mathcal{D},\mathbf{Set}]$ (see e.g.~\cite[VII.2]{mm}). The inverse image functor $\Phi^{\ast}$ is given by
\begin{equation*}
\forall \uA\in[\mathcal{D},\mathbf{Set}],\ \ \forall C\in\mathcal{C}\ \ \Phi^{\ast}(\uA)(C)=\uA(\phi(C)).
\end{equation*}
This geometric morphism, $\Phi$, has the property that $\Phi^{\ast}$ has a left adjoint $\Phi_{!}$. A geometric morphism with this property is called an \emph{essential} geometric morphism.
\end{exa}

For any geometric morphism $F:\mathcal{E}\to\mathcal{F}$, the inverse image functor $F^{\ast}$ preserves all colimits (it is a right adjoint) and all finite limits (it is left exact). This has important consequences. The functor $F^{\ast}$ preserves certain objects, such as the terminal object $F^{\ast}(\underline{1}_{\mathcal{F}})\cong\underline{1}_{\mathcal{E}}$, the natural numbers object $F^{\ast}(\underline{\mathbb{N}}_{\mathcal{F}})\cong\underline{\mathbb{N}}_{\mathcal{E}}$, and the objects $\qup$, and $\qui$.

If $\uA$ is a C*-algebra in $\mathcal{F}$, with addition $+:\uA\times\uA\to\uA$, then by left-exactness $F^{\ast}(+)$ corresponds to an arrow $F^{\ast}\uA\times F^{\ast}\uA\to F^{\ast}\uA$. In this way, $F^{\ast}\uA$ obtains arrows 
\begin{equation*}
\cdot:\qui\times F^{\ast}\uA\to F^{\ast}\uA,\ \  0:\underline{1}\to F^{\ast}\uA,\ \ \text{etc}\ldots
\end{equation*}
Many of the axioms for a C*-algebra hold for $F^{\ast}\uA$, as $F^{\ast}$ preserves all axioms which are expressed using geometric logic. This means that the axiom is of the form
\begin{equation*}
\forall x_{1}\in\underline{A}_{1},\ldots,\forall x_{n}\in\underline{X}_{n}\ \ \phi(x_{1},\ldots,x_{n})\to\psi(x_{1},\ldots,x_{n}),
\end{equation*}
where $\phi(x_{1},\ldots,x_{n})$ and $\psi(x_{1},\ldots,x_{n})$ are geometric formulae. A formula is geometric if it is constructed from the variables $x_{1},\ldots,x_{n}$, terms and predicates, using equalities $=$, finite conjunctions $\wedge$, arbitrary disjunctions $\bigvee$, and existential quantifiers $\exists$. A theory is called geometric if all axioms are of the form given above. Categorically, this means that the formulae are constructed using finite limits and colimits, which are constructions preserved by $F^{\ast}$. As an example, consider the axiom
\begin{equation*}
\forall a\in\uA\ \exists p\in\qup\ \ (a,p)\in\uN.
\end{equation*}
This axiom is equivalent to
\begin{equation*}
\forall a\in\uA\ \ \left(\left(\exists p\in\qui\ (a,p)\in\uN\right)\to\top\right),
\end{equation*}
which is of the desired form (note that $\top$ is allowed in geometric logic, as it is the conjunction over the empty set). The inverse image functor $F^{\ast}$ respects this axiom in the sense that the axiom
\begin{equation*}
\forall a\in F^{\ast}\uA\ \exists p\in\qup\ \ (a,p)\in F^{\ast}\uN
\end{equation*}
holds for $F^{\ast}\uA$. For a proper discussion of geometric logic and its relation to geometric morphisms, see~\cite{vic}, \cite[Chapter X]{mm}, \cite[D1]{jh1}.

\section{Nuiten's sheaves} \label{sec: sheaf}

This section reviews the sheaf condition as introduced by Joost Nuiten in his impressive bachelor thesis~\cite{nuiten}, supervised by Dr. Urs Schreiber. This is a stepping stone to the C*-algebraic version treated in Section~\ref{sec: pullback}, which uses the category $\mathbf{ucCSp}$ instead of $\mathbf{RingSp}$, whilst also motivating the notion of strong locality, the central concept of Section~\ref{sec: locality}. Note that in~\cite{nuiten}, additivity of the net of operator algebras is assumed at certain places; our treatment below shows that we do not need to assume this.

In the covariant topos model, to a unital C*-algebra $A$ in $\mathbf{Set}$ we associate a topos $\mathcal{T}_{A}=[\mathcal{C}_{A},\mathbf{Set}]$, as well as a unital commutative C*-algebra $\underline{A}$ in this topos.  Any $\ast$-homomorphism $f:A\to B$ induces an order-preserving function
\begin{equation*}
\hat{f}:\mathcal{C}_{A}\to\mathcal{C}_{B},\ \ \hat{f}(C)=f[C].
\end{equation*}
In turn, this function induces an essential geometric morphism $F:\mathcal{T}_{A}\to\mathcal{T}_{B}$. In this way, we obtain a functor from the category of C*-algebras and $\ast$-homomorphisms in $\mathbf{Set}$ to the category of Grothendieck topoi and geometric morphisms. However, if we restrict to $\ast$-homomorphisms that reflect commutativity in the sense that
\begin{equation*}
\forall a,b\in A,\ \ [f(a),f(b)]=0\Rightarrow [a,b]=0,
\end{equation*}
such as embeddings of C*-algebras, we can do the following. For a commutativity reflecting $f$, the inverse image map defines an order-preserving function
\begin{equation*}
f^{-1}:\mathcal{C}_{B}\to\mathcal{C}_{A},\ \ C\mapsto f^{-1}(C),
\end{equation*}
which induces an essential geometric morphism $G:\mathcal{T}_{B}\to\mathcal{T}_{A}$, and defines a contravariant functor from the category of C*-algebras and $\ast$-homomorphisms that reflect commutativity to the category of Grothendieck topoi and geometric morphisms.

\begin{lem}
Let $\psi:A\to B$ is a unit-preserving $\ast$-homomorphism. 
\begin{enumerate}
\item If $ker(\psi)=\{0\}$, then $\psi$ reflects commutativity. 
\item If $A$ is simple, then $\psi$ reflects commutativity.
\item If $A=\mathcal{B}(\mathcal{H})$, for a separable Hilbert space $\mathcal{H}$, then $\psi$ reflects commutativity iff $ker(\psi)=\{0\}$.
\end{enumerate}
\end{lem}

\begin{proof}
The first claim follows from
\begin{equation*}
0=[\psi(a),\psi(b)]=\psi([a,b])\ \ \Rightarrow\ \ [a,b]\in ker(\psi)=\{0\}.
\end{equation*}
For a simple C*-algebra $A$, the kernel of a $\ast$-homomorphism $\psi:A\to B$ is either $A$ or $\{0\}$. As $\psi$ preserves the unit, it follows that $ker(\psi)=\{0\}$, and $\psi$ reflects commutativity.

Next, let $A=\mathcal{B}(\mathcal{H})$. Assume that $\psi$ reflects commutativity, and let $a\in ker(\psi)$. As $\psi(a)=0$, $\psi(a)$ commutes with $\psi(b)$, for each $b\in\mathcal{B}(\mathcal{H})$. By assumption, $a$ commutes with all elements of $\mathcal{B}(\mathcal{H})$. As the centre of $\mathcal{B}(\mathcal{H})$ is $\mathbb{C}1$, we deduce $a=\lambda1$ for some $\lambda\in\mathbb{C}$. By assumption $0=\psi(\lambda1)=\lambda\psi(1)=\lambda$. Consequently $a=0$, and $ker(\psi)=\{0\}$.
\end{proof}

\begin{cor}
For a a pure state $\rho$ of $\mathcal{B}(\mathcal{H})$, the corresponding GNS representation $\pi_{\rho}:\mathcal{B}(\mathcal{H})\to\mathcal{B}(\mathcal{H}_{\rho})$ reflects commutativity iff $\rho$ is normal. 
\end{cor}

\begin{proof}
If $\rho$ is normal, then $\pi_{\rho}$ is faithful, and therefore reflects commutativity. If $\rho$ is not normal, then by~\cite[Theorem 10.4.6]{kari2}, $\pi_{\rho}(\mathcal{K})=\{0\}$, where $\mathcal{K}$ denotes the ideal of compact operators. Clearly, $\pi_{\rho}$ does not reflect commutativity in this case.
\end{proof}

As another example of a $\ast$-homomorphism that does not reflect commutativity, consider a continuous field of C*-algebras $(A,\{A_{x},\psi_{x}\}_{x\in X})$ (see e.g.~\cite[Chapter 10]{dixmier}). Here $A$ is a C*-algebra, $X$ a locally compact Hausdorff space, and for each $x\in X$ we are given a surjective $\ast$-homomorphism $\psi_{x}:A\to A_{x}$. If for $a\in A$ we define $a(x):=\psi_{x}(a)$, then $a$ can be identified with the family $\{a(x)\}_{x\in X}$. Note that
\begin{equation*}
[a,b]_{A}=0\ \ \Leftrightarrow \ \ \forall x\in X\ [a(x),b(x)]_{A_{x}}=0.
\end{equation*}
If there exists an $y\in X$ such that $A_{y}$ is commutative, but at least one $A_{x}$ is non-commutative, then the $\ast$-homomorphism $\psi_{y}:A\to A_{y}$ does not reflect commutativity.

After these remarks on reflection of commutativity, we return to the discussion of associating geometric morphisms to $\ast$-homomorphisms. Viewing this process as a contravariant functor has the advantage that it does not only define a geometric morphism, but a morphism of ringed topoi as well.

\begin{dork}
The category $\mathbf{RingTopos}$ of ringed topoi is given by:
\begin{itemize}
\item Objects are pairs $(\mathcal{E},\underline{R})$, with $\mathcal{E}$ a topos, and $\underline{R}$ a commutative ring with unit, internal to $\mathcal{E}$.
\item An arrow $(F,\underline{f}):(\mathcal{E},\underline{R})\to(\mathcal{F},\underline{S})$ is given by a geometric morphism $F:\mathcal{E}\to\mathcal{F}$, and a ring homomorphism $\underline{f}:F^{\ast}\underline{S}\to\underline{R}$ in $\mathcal{E}$.
\item Composition is defined by $(G,\underline{g})\circ(F,\underline{f})=(G\circ F,\underline{f}\circ F^{\ast}\underline{g})$.
\end{itemize}
\end{dork}

For the morphism $(G,\underline{g})$ induced by a commutativity-reflecting $\ast$-homomorphism $f:A\to B$, the inverse image functor $G^{\ast}:\mathcal{T}_{A}\to\mathcal{T}_{B}$ is easily described. If $\underline{F}\in\mathcal{T}_{A}$ and $D\in\mathcal{C}_{B}$, then $G^{\ast}(\underline{F})(D)=\underline{F}(f^{-1}(D))$. For a geometric morphism induced by a functor on the base categories, the direct image is usually harder to describe. However, in the case at hand we can easily describe $G_{\ast}$. Let $F:\mathcal{T}_{A}\to\mathcal{T}_{B}$ be the essential geometric morphism induced by $\hat{f}:\mathcal{C}_{A}\to\mathcal{C}_{B}$. Then $G_{\ast}=F^{\ast}$. So if $\underline{F}\in\mathcal{T}_{B}$ and $C\in\mathcal{C}_{A}$, then $G_{\ast}(\underline{F})(C)=\underline{F}(f[C])$. The ring part of the morphism of ringed topoi, $\underline{g}:G^{\ast}\underline{A}\to\underline{B}$, is the natural transformation given by $\underline{g}_{D}:f^{-1}(D)\to D$ being the restriction of $f$ to $f^{-1}(D)$, 
\begin{equation*}
\underline{g}_{D}=f|_{f^{-1}(D)}.
\end{equation*}
From the discussion above, it is clear that we obtained a contravariant functor from the category unital C*-algebras and unit-preserving $\ast$-homomorphisms that reflect commutativity to the category $\mathbf{RingTopos}$. Instead of using the category $\mathbf{RingTopos}$, we now restrict to the subcategory $\mathbf{RingSp}$. This makes it easier to calculate limits later on.

\begin{dork}
The category $\mathbf{RingSp}$ of ringed spaces is the following subcategory of $\mathbf{RingTopos}$:
\begin{itemize}
\item Objects are pairs $(X,\underline{R})$, with $X$ a topological space, and $\underline{R}$ a commutative ring with unit internal to $Sh(X)$.
\item An arrow $(f,\underline{f}):(X,\underline{R})\to(Y,\underline{S})$ is given by a continuous map $f:X\to Y$, and a ring homomorphism $\underline{f}:F^{\ast}\underline{S}\to\underline{R}$ in $Sh(X)$, where $F:Sh(X)\to Sh(Y)$ is the geometric morphism induced by $f$.
\end{itemize}
With slight abuse of notation, we will also write $(Sh(X),\underline{R})$ for the object $(X,\underline{R})$, as well as $(F,\underline{f})$ for an arrow $(f,\underline{f})$, emphasising that $\mathbf{RingSp}$ is indeed a subcategory of $\mathbf{RingTopos}$.
\end{dork}

If $P$ is a poset, then $P$ can be seen as a topological space $P_{\uparrow}$ by equipping it with the Alexandroff (upper set) topology, defined as
\begin{equation*}
U\in\mathcal{O}P\ \ \leftrightarrow\ \ \forall p\in P\ \ (p\in U)\wedge(p\leq q)\rightarrow(q\in U).
\end{equation*}
If we identify the elements $p\in P$ with the Alexandroff opens $(\uparrow p)\in\mathcal{O}P_{\uparrow}$, the topos $Sh(P_{\uparrow})$ is isomorphic to the topos $[P,\mathbf{Set}]$. This implies that for any C*-algebra $A$, the pair $(\mathcal{T}_{A},\underline{A})$ lies in $\mathbf{RingSp}$. Any order-preserving map $P\to Q$ of posets is an Alexandroff continuous map. A straightforward check then reveals that the geometric morphism $G: Sh(\mathcal{C}_{B})\to Sh(\mathcal{C}_{A})$ induced by the continuous map $f^{-1}:\mathcal{C}_{B}\to\mathcal{C}_{A}$, is, under the identification $Sh(\mathcal{C})\cong[\mathcal{C},\mathbf{Set}]$, the same geometric morphism as the one induced by $f^{-1}$ seen as a functor on poset categories. The morphisms of ringed topoi induced by $\ast$-homomorphisms are present in $\mathbf{RingSp}$ as well.

Now that we have defined the category of interest, we move to AQFT and derive Nuiten's sheaf condition. Consider the following situation. Suppose we are given a net $\mathcal{O}\mapsto A(\mathcal{O})$ of operator algebras. Throughout the paper we assume that for each region of spacetime $\mathcal{O}$ under consideration, the C*-algebra $A(\mathcal{O})$ is unital. For the moment, the only other assumption on the net is isotony, i.e., if $O_{1}\leq O_{2}$, then $A(O_{1})\subseteq A(O_{2})$.

Let $\mathcal{V}(X)$ denote a poset of certain causally complete opens of a spacetime manifold $X$, partially ordered by inclusion. As before, the details of $\mathcal{V}(X)$ are unimportant, but we will assume that the poset has binary joins and meets. Let $A:\mathcal{V}(X)\to\mathbf{CStar_{rc}}$ be a net of C*-algebras. The subscript \textit{rc} means that we restrict ourselves to morphisms that reflect commutativity. By isotony, the maps $A(O_{1})\to A(O_{2})$, corresponding to inclusions $O_{1}\subseteq O_{2}$, are inclusion maps, clearly satisfying the constraint of reflecting commutativity. We therefore obtain a contravariant functor\footnote{In~\cite{doering} D\"oring presents a slightly different way of dealing with $\ast$-homomorphisms. The arrows introduced by D\"oring are closely related to those of Nuiten~\cite{wollie2}. Using the version by D\"oring yields a covariant functor $A:\mathcal{V}(X)\to\mathbf{RingSp}$, instead of a contravariant functor.} 
\begin{equation*}
\underline{A}: \mathcal{V}(X)^{\text{op}}\to\mathbf{RingSp},\ \ \underline{A}(O)=(\mathcal{T}_{A(O)},\underline{A(O)}),
\end{equation*}
where the inclusion $O_{1}\subseteq O_{2}$ is mapped to $\underline{A}(O_{1}\subseteq O_{2})=(I,\underline{i})$, with 
\begin{equation*}
I^{\ast}:\mathcal{T}_{A(O_{1})}\to\mathcal{T}_{A(O_{2})}\ \ I^{\ast}(\underline{F})(C)=\underline{F}(C\cap A(O_{1})),
\end{equation*}
\begin{equation*}
\underline{i}:I^{\ast}(\underline{A(O_{1})})\to\underline{A(O_{2})},\ \ \underline{i}_{C}:C\cap A(O_{1})\hookrightarrow C,
\end{equation*}
and where the ring morphisms are inclusion maps. 

Let $O_{1},O_{2}\in\mathcal{V}(X)$. It will be convenient to introduce the following notation 
\begin{equation*}
A_{i}:=A(O_{i}),\ \ \ A_{1\wedge2}:=A(O_{1}\wedge O_{2}),\ \ \ A_{1\vee2}:=A(O_{1}\vee O_{2}),
\end{equation*}
\begin{equation*}
(\mathcal{T}_{i},\underline{A}_{i}):=(\mathcal{T}_{A_{i}},\underline{A_{i}}),
\end{equation*}
\begin{equation*}
(\mathcal{T}_{1\wedge2},\underline{A}_{1\wedge2}):=(\mathcal{T}_{A_{1\wedge2}},\underline{A_{1\wedge2}}),
\end{equation*}
\begin{equation*}
(\mathcal{T}_{1\vee2},\underline{A}_{1\vee2}):=(\mathcal{T}_{A_{1\vee2}},\underline{A_{1\vee2}}).
\end{equation*}
Consider the following diagram in $\mathbf{RingSp}$, where the morphisms $(I_{i},\underline{i}_{i})$ are induced by the inclusions $O_{1}\wedge O_{2}\subseteq O_{i}$ and the morphisms $(J_{i},\underline{j}_{i})$ are induced by the inclusions $O_{i}\subseteq O_{1}\vee O_{2}$.
\[ \xymatrix{
(\mathcal{T}_{1\vee2},\underline{A}_{1\vee2}) \ar@/_/[dddr]_{(J_{2},\underline{j}_{2})} \ar@/^/[drrr]^{(J_{1},\underline{j}_{1})} \ar@{.>}[dr]|-{(H,\underline{h})} \\
&(\mathcal{F},\underline{R}) \ar[dd]^{(P_{2},\underline{p}_{2})} \ar[rr]_{(P_{1},\underline{p}_{1})} & & (\mathcal{T}_{1},\underline{A}_{1})\ar[dd]_{(I_{1},\underline{i}_{1})} \\
\\
&(\mathcal{T}_{2},\underline{A}_{2}) \ar[rr]^{(I_{2},\underline{i}_{2})} & &(\mathcal{T}_{1\wedge2},\underline{A}_{1\wedge2}).} \]
The bottom square of the diagram is a pullback. As the category $\mathbf{RingSp}$ is complete, this pullback exists, and we will compute it below. We will think of the pullback object $(\mathcal{F},\underline{R})$ as the ringed topos of matching families for the cover $\{O_{1}, O_{2}\}$ of $O_{1}\vee O_{2}$. We are now ready to formulate Nuiten's sheaf condition.

\begin{dork} \label{def: sheafcond}
The functor $A:\mathcal{V}(X)^{op}\to\mathbf{RingSp}$ is said to be a sheaf iff  for each pair $O_{1},O_{2}\in\mathcal{V}(X)$ of spacelike separated opens, the descent morphism
\begin{equation*}
(H,\underline{h}):(\mathcal{T}_{1\vee2},\underline{A}_{1\vee2})\to(\mathcal{F},\underline{R}).
\end{equation*}
is an isomorphism of ringed spaces. 
\end{dork}

Let us briefly compare this sheaf condition with the sheaf condition used for topoi $Sh(X)$, where $X$ is a topological space. Let $F:\mathcal{O}X^{op}\to\mathbf{Set}$ be a presheaf, and $U\in\mathcal{O}X$ an open subset covered by smaller open subsets $\{U_{i}\}_{i\in I}$, in the sense of $U=\bigcup_{i\in I}U_{i}$. Consider the equalizer
\[\xymatrix{E\ \ar@{^{(}->}[r] & \prod_{i\in I}F(U_{i}) \ar@/^/[rr]^{p} \ar@/_/[rr]_{q} & & \prod_{i\neq j}F(U_{i}\wedge U_{j})},\]
where
\begin{equation*}
p((f_{k})_{k\in I})_{ij}:=f_{i}|_{U_{i}\wedge U_{j}},\ \ \ q((f_{k})_{k\in I})_{ij}:=f_{j}|_{U_{i}\wedge U_{j}}.
\end{equation*}
The presheaf $F$ is a sheaf iff for each such $U$ and $\{U_{i}\}_{i\in I}$, the descent morphism
\begin{equation*}
F(U)\to E,\ \ \ f\mapsto(f|_{U_{i}})_{i\in I},
\end{equation*}
is an isomorphism. Note that we can replace $\mathbf{Set}$ by any complete category, such as $\mathbf{RingSp}$, leading to the sheaf condition of the previous definition.\\

The next step is to make the descent morphism $(H,\underline{h})$ explicit in order to understand the sheaf condition at the level of the net $A:\mathcal{V}(X)\to\mathbf{CStar}$, and to investigate if this mathematically sensible condition is plausible on physical grounds as well. We start by finding the space $X$ of the topos $\mathcal{F}=Sh(X)$. The geometric morphisms $I_{i}$ and $J_{i}$ are induced by order-preserving functions
\begin{equation*}
y_{i}:\mathcal{C}_{1\vee2}\to\mathcal{C}_{i},\ \ \ y_{i}(C)=C\cap A_{i},
\end{equation*}
\begin{equation*}
x_{i}:\mathcal{C}_{i}\to\mathcal{C}_{1\wedge2},\ \ \ x_{i}(C)=C\cap A_{1\wedge2},
\end{equation*}
where we used the notation $\mathcal{C}_{i}:=\mathcal{C}(A(\mathcal{O}_{i}))$, etc$\ldots$ Define the poset
\begin{equation*}
\mathcal{C}_{1}\times_{\mathcal{C}_{1\wedge2}}\mathcal{C}_{2}=\{(C_{1},C_{2})\in\mathcal{C}_{1}\times\mathcal{C}_{2}\mid C_{1}\cap A_{1\wedge2}=C_{2}\cap A_{1\wedge2}\},
\end{equation*}
with partial order $(D_{1},D_{2})\leq (C_{1},C_{2})$ iff $D_{1}\subseteq C_{1}$ and $D_{2}\subseteq C_{2}$, and (order-preserving) projection maps $\pi_{i}:\mathcal{C}_{1}\times_{\mathcal{C}_{1\wedge2}}\mathcal{C}_{2}\to\mathcal{C}_{i}$. In the category $\mathbf{Poset}$ we obtain the pullback square
\[\xymatrix{ \mathcal{C}_{1}\times_{\mathcal{C}_{1\wedge2}}\mathcal{C}_{2}  \ar[d]^{\pi_{2}} \ar[r]_{\pi_{1}} & \mathcal{C}_{1}\ar[d]_{x_{1}} \\
\mathcal{C}_{2} \ar[r]^{x_{2}} &\mathcal{C}_{1\wedge2}\ \ .} \]
Taking the Alexandroff upper topology of a poset defines a functor $Al:\mathbf{Poset}\to\mathbf{Top}$, where $\mathbf{Top}$ is the category of topological spaces and continuous maps. This functor preserves limits~\cite{nuiten}. Note that if we replaced $\mathbf{Top}$ by the category of locales or topoi, then the functor $Al$ would not preserve all limits. With respect to the Alexandroff upper topologies, the previous square becomes a pullback in $\mathbf{Top}$. It will turn out that $\mathcal{C}_{1}\times_{\mathcal{C}_{1\wedge2}}\mathcal{C}_{2}$, equipped with the Alexandroff upper topology, is the space we are looking for. Once this has been shown, we will conclude that $\mathcal{F}=[\mathcal{C}_{1}\times_{\mathcal{C}_{1\wedge 2}}\mathcal{C}_{2},\mathbf{Set}]$. Let
\begin{equation*}
P_{i}:[\mathcal{C}_{1}\times_{\mathcal{C}_{1\wedge 2}}\mathcal{C}_{2},\mathbf{Set}]\to[\mathcal{C}_{i},\mathbf{Set}]
\end{equation*}
denote the geometric morphisms corresponding to the projections $\pi_{i}$. The next step in describing the descent morphism is to compute the following pushout of rings in $\mathcal{F}$:
\[ \xymatrix{
P_{1}^{\ast}I^{\ast}_{1}\underline{A}_{1\wedge2} \ar[r]^{P_{1}^{\ast}\underline{i}_{1}} \ar[d]_{P_{2}^{\ast}\underline{i}_{2}} & P^{\ast}_{1}\underline{A}_{1} \ar[d]^{\underline{p}_{1}} \\
P_{2}^{\ast}\underline{A}_{2} \ar[r]_{\underline{p}_{2}} & \underline{R}\ \ ,} \]
where we used $P^{\ast}_{1}I^{\ast}_{1}=P_{2}^{\ast}I^{\ast}_{2}$. 

In a functor category $[\mathcal{C},\mathbf{Set}]$, an object $\underline{R}$ is an internal ring iff it is a functor $\underline{R}:\mathcal{C}\to\mathbf{Ring}$. This entails that we can compute the pushout $\underline{R}$ stage-wise. Taking $(C_{1},C_{2})\in\mathcal{C}\times_{\mathcal{C}_{1\wedge2}}\mathcal{C}_{2}$, we compute the pushout of rings in $\mathbf{Set}$. Using
\begin{equation*}
P^{\ast}_{1}I^{\ast}_{1}\underline{A}_{1\wedge2}(C_{1},C_{2})=C_{1}\cap A_{1\wedge2}=C_{2}\cap A_{1\wedge2}=P^{\ast}_{2}I^{\ast}_{2}\underline{A}_{1\wedge2}(C_{1},C_{2}),
\end{equation*}
and $P^{\ast}_{1}\underline{A}_{1}(C_{1},C_{2})=C_{1}$ and $P^{\ast}_{2}\underline{A}_{2}(C_{1},C_{2})=C_{2}$, we obtain the pushout square
\[ \xymatrix{
C_{1}\cap A_{1\wedge2} \ar[r] \ar[d] & C_{1} \ar[d]^{-\otimes1} \\
C_{2} \ar[r]_{1\otimes-} & C_{1}\otimes_{C_{1}\cap A_{1\wedge2}}C_{2}\ \ ,} \]
where the unlabelled arrows are inclusion maps, we used that $C_{1}\cap A_{1\wedge2}=C_{2}\cap A_{1\wedge2}$, and we used that for commutative rings, the pushout ring is given by the tensor product of $C_{1}$ and $C_{2}$, viewed as $C_{1}\cap C_{2}\cap A_{1\wedge2}$-algebras.

\begin{lem}
Define $(\mathcal{F},\underline{R})\in\mathbf{RingSp}$ as $\mathcal{F}=[\mathcal{C}_{1}\times_{\mathcal{C}_{1\wedge2}}\mathcal{C}_{2},\mathbf{Set}]$ and put
\begin{equation*}
\underline{R}:\mathcal{C}_{1}\times_{\mathcal{C}_{1\wedge2}}\mathcal{C}_{2}\to\mathbf{Set},\ \ \underline{R}(C_{1},C_{2})=C_{1}\otimes_{C_{12}}C_{2},
\end{equation*}
where we used the notation $C_{12}=C_{1}\cap C_{2}\cap A_{1\wedge2}$. If $(D_{1},D_{2})\leq(C_{1},C_{2})$ in $\mathcal{C}_{1}\times_{\mathcal{C}_{1\wedge2}}\mathcal{C}_{2}$, the corresponding ring homomorphism is simply
\begin{equation*}
\underline{R}(\leq):D_{1}\otimes_{D_{12}}D_{2}\to C_{1}\otimes_{C_{12}}C_{2}\ \ a\otimes b\mapsto a\otimes b.
\end{equation*}
Define $\underline{p}_{i}:P_{i}^{\ast}\underline{A}_{i}\to\underline{R}$ as
\begin{equation*}
(\underline{p}_{1})_{(C_{1},C_{2})}:C_{1}\to C_{1}\otimes_{C_{12}}C_{2}\ \ a\mapsto a\otimes1,
\end{equation*}
\begin{equation*}
(\underline{p}_{2})_{(C_{1},C_{2})}:C_{2}\to C_{1}\otimes_{C_{12}}C_{2}\ \ b\mapsto 1\otimes b.
\end{equation*}
Then the following diagram is a pullback in $\mathbf{RingSp}$:
\[ \xymatrix{
(\mathcal{F},\underline{R}) \ar[d]_{(P_{2},\underline{p}_{2})} \ar[r]^{(P_{1},\underline{p}_{1})} & (\mathcal{T}_{1},\underline{A}_{1})\ar[d]^{(I_{1},\underline{i}_{1})} \\
(\mathcal{T}_{2},\underline{A}_{2}) \ar[r]_{(I_{2},\underline{i}_{2})} &(\mathcal{T}_{1\wedge2},\underline{A}_{1\wedge2})\ \ .} \]
\end{lem}

\begin{proof}
Suppose that we have the following commutative diagram in $\mathbf{RingSp}$:
\[ \xymatrix{
(Sh(X),\underline{S}) \ar@/_/[dddr]_{(F_{2},\underline{f}_{2})} \ar@/^/[drrr]^{(F_{1},\underline{f}_{1})} \ar@{.>}[dr]|-{(H,\underline{h})} \\
&(\mathcal{F},\underline{R}) \ar[dd]^{(P_{2},\underline{p}_{2})} \ar[rr]_{(P_{1},\underline{p}_{1})} & & (\mathcal{T}_{1},\underline{A}_{1})\ar[dd]_{(I_{1},\underline{i}_{1})} \\
\\
&(\mathcal{T}_{2},\underline{A}_{2}) \ar[rr]^{(I_{2},\underline{i}_{2})} & &(\mathcal{T}_{1\wedge2},\underline{A}_{1\wedge2})\ \ .} \]
We need to show that there exists a unique $(H,\underline{h})$ completing the diagram. By definition of $\mathcal{F}$, there exists a unique continuous map $h:X\to\mathcal{C}_{1}\times_{\mathcal{C}_{1\wedge2}}\mathcal{C}_{2}$ such that
\[ \xymatrix{
X \ar@/_/[ddr]_{f_{2}} \ar@/^/[drr]^{f_{1}} \ar@{.>}[dr]|-{h} \\
& \mathcal{C}_{1}\times_{\mathcal{C}_{1\wedge2}}\mathcal{C}_{2} \ar[d]^{\pi_{2}} \ar[r]_{\pi_{1}} & \mathcal{C}_{1}\ar[d]_{i_{1}} \\
&\mathcal{C}_{2} \ar[r]^{i_{2}} &\mathcal{C}_{1\wedge2}} \]
is a commutative diagram. Let $H$ be the geometric morphism corresponding to $h$. For the next step we consider the action of the inverse image functor $H^{\ast}$ on pushout diagrams of rings in $\mathcal{F}$.

If $F:\mathcal{E}\to\mathcal{F}$ is any geometric morphism, then $F^{\ast}$ will map a pushout square of rings in $\mathcal{F}$ to a pushout square of rings in $\mathcal{E}$. This can be verified in a straightforward way using naturality of the adjunction $F^{\ast}\vdash F_{\ast}$, and the fact that $F_{\ast}$ is left-exact. As a consequence, for a ring $\underline{R}$ in $\mathcal{E}$, the object $F_{\ast}\underline{R}$ is a ring in $\mathcal{F}$. An arrow $F^{\ast}\underline{S}\to\underline{R}$ is a ring homomorphism in $\mathcal{E}$ iff the corresponding arrow $\underline{S}\to F_{\ast}\underline{R}$ is a ring homomorphism in $\mathcal{F}$\ \footnote{Note that if we replace the algebraic theory of rings by a more general geometric theory, this need not be the case. One reason is that for a model $\underline{R}$ of a geometric theory, the object $F_{\ast}\underline{R}$ need not be a model of the same theory, as $F_{\ast}$ need not preserve any colimits.}. 

By the previous considerations, we know that the square below is a pushout of rings in $Sh(X)$.
\[ \xymatrix{
H^{\ast}P_{1}^{\ast}I^{\ast}_{1}\underline{A}_{1\wedge2} \ar[r]^{H^{\ast}P_{1}^{\ast}\underline{i}_{1}} \ar[d]_{H^{\ast}P_{2}^{\ast}\underline{i}_{2}} & H^{\ast}P^{\ast}_{1}\underline{A}_{1} \ar[d]^{H^{\ast}\underline{p}_{1}} \ar@/^/[ddr]^{\underline{f}_{1}} \\
H^{\ast}P_{2}^{\ast}\underline{A}_{2} \ar@/_/[drr]_{\underline{f}_{2}} \ar[r]_{H^{\ast}\underline{p}_{2}} & H^{\ast}\underline{R} \ar[dr]^{\underline{h}}\\
& &  \underline{S}\ \ .} \]
The pair $(H,\underline{h})$ exists and is unique.
\end{proof}

Using this lemma, we can write down an explicit expression for the descent morphism.

\begin{lem}
The descent morphism is given by
\begin{equation*}
(H,\underline{h}):([\mathcal{C}_{1\vee2},\mathbf{Set}],\underline{A}_{1\vee2})\to([\mathcal{C}_{1}\times_{\mathcal{C}_{1\wedge2}}\mathcal{C}_{2},\mathbf{Set}],\underline{R}),
\end{equation*}
where $H$ is the geometric morphism induced by the poset map
\begin{equation*}
h:\mathcal{C}_{1\vee2}\to\mathcal{C}_{1}\times_{\mathcal{C}_{1\wedge2}}\mathcal{C}_{2},\ \ C\mapsto(C\cap A_{1},C\cap A_{2}),
\end{equation*}
and the ring morphism $\underline{h}:H^{\ast}\underline{R}\to\underline{A}_{1\vee 2}$ in $[\mathcal{C}_{1\vee2},\mathbf{Set}]$ is given by the functions
\begin{equation} \label{equ: xi}
\underline{h}_{C}:(C\cap A_{1})\otimes_{C\cap A_{1\wedge2}}(C\cap A_{2})\to C,\ \ a\otimes b\mapsto a\cdot b.
\end{equation}
\end{lem}

Note that (\ref{equ: xi}) follows from
\begin{equation*}
\underline{h}_{C}(a\otimes b)=\underline{h}_{C}(a\otimes1\cdot 1\otimes b)=\underline{h}_{C}(a\otimes1)\cdot\underline{h}_{C}(1\otimes b)=a\cdot b.
\end{equation*}

In order to be a sheaf, by definition the morphism $(H,\underline{h})$ must be an isomorphism in $\mathbf{RingSp}$. In particular, $h$ has to be a homeomorphism and therefore a bijection. This seems like a strong demand, as shown by the following example. 

\begin{exa}
Let $A_{1\vee2}=C([0,1]^{2})$ be the C*-algebra of continuous complex-valued functions on the unit square. Let
\begin{equation}
A_{1}=\{f\in C([0,1]^{2})\mid \exists g\in C([0,1]),\  \forall x,y,\ f(x,y)=g(x)\},
\end{equation}
\begin{equation}
A_{2}=\{f\in C([0,1]^{2})\mid \exists g\in C([0,1]),\  \forall x,y,\ f(x,y)=g(y)\},
\end{equation}
Note that $A_{1\vee2}=A_{1}\otimes A_{2}$ as C*-algebras in this example. In particular, $A_{1}\cap A_{2}=\mathbb{C}$. Consider $C\in\mathcal{C}_{1\vee2}$ given by
\begin{equation} \label{countersheaf}
C=\{f\in C([0,1]^{2})\mid\forall x\in[0,1],\ \ f(x,x)=f(0,0)\}.
\end{equation}
Clearly, $C\cap A_{1}=C\cap A_{2}=\mathbb{C}$. Consequently, $h(C)=h(\mathbb{C})$ hence $h$ is not injective, so that it does not define an isomorphism of posets or spaces. Note that in this example $\underline{h}_{C}:\mathbb{C}\to C$ is the inclusion map, which is a ring morphism that is not surjective.
\end{exa}

\begin{exa} \label{exa: counter}
We can simplify the previous example in order to demonstrate that the full sheaf condition can be expected to fail for physically reasonable nets. Let $2=\{0,1\}$ be the two element discrete space. Define $A\vee B\cong A\otimes B\cong C(2\times 2)$.
\begin{equation*}
A=\{f:2\times2\to\mathbb{C}\mid f(0,0)=f(0,1),\ f(1,0)=f(1,1)\}\cong C(2),
\end{equation*}
\begin{equation*}
B=\{f:2\times2\to\mathbb{C}\mid f(0,0)=f(1,0),\ f(0,1)=f(1,1)\}\cong C(2).
\end{equation*}
Let $e_{ij}=e_{i}\otimes e_{j}$ denote the characteristic function
\begin{equation*}
e_{ij}(k,l)=\delta_{ik}\delta_{jl},\ \ i,j,k,l\in\{0,1\}.
\end{equation*}
Consider the unital subalgebra $C$ of $A\otimes B$ generated by $e_{10}-e_{01}$. The C*-algebra $C$ consists of functions $f:2\times 2\to\mathbb{C}$ of the form
\begin{equation*}
f=\alpha_{0}1+\alpha_{1}(e_{10}-e_{01})+\alpha_{2}(e_{10}+e_{01}),\ \ \alpha_{0},\alpha_{1},\alpha_{2}\in\mathbb{C},
\end{equation*}
where $1$ denotes the constant function. Note that $C\cap A=C\cap B=\mathbb{C}$. Consequently, $h(C)=h(\mathbb{C})$ and the sheaf condition does not hold.
\end{exa}

If the full sheaf condition is too strong, we could consider weaker versions instead. Nuiten introduces strong locality as such an alternative. However, we first consider microcausality. Microcausality is the assumption that if $O_{1}$ and $O_{2}$ are spacelike separated, then $[A_{1},A_{2}]=\{0\}$. This condition may be reformulated quite elegantly as

\begin{poe}{(Nuiten's Lemma~\cite{nuiten})} \label{poe: micro}
Microcausality is equivalent to the property that the poset morphism $h:\mathcal{C}_{1\vee2}\to\mathcal{C}_{1}\times_{\mathcal{C}_{1\wedge2}}\mathcal{C}_{2}$ has a left adjoint $\vee$. 
\end{poe}

\begin{proof}
If we assume microcausality and $(C_{1},C_{2})\in\mathcal{C}_{1}\times_{\mathcal{C}_{1\wedge2}}\mathcal{C}_{2}$, then $C_{1}\cup C_{2}$ is commutative in $A_{1\vee 2}$ and generates a context $\vee(C_{1},C_{2})$ in $\mathcal{C}_{1\vee2}$, which we denote by $C_{1}\vee C_{2}$. By construction
\begin{equation} \label{equ: adj}
C_{1}\vee C_{2}\subseteq C\ \ \text{iff}\ \ (C_{1}\subseteq C\cap A_{1})\ \text{and}\ (C_{2}\subseteq C\cap A_{2}).
\end{equation}
Conversely, assume that $h$ has a left adjoint $\vee$. By setting $C=C_{1}\vee C_{2}$ in (\ref{equ: adj}) we find $C_{1}, C_{2}\subseteq(C_{1}\vee C_{2})$. As $C_{1}\vee C_{2}$ is commutative, $[C_{1},C_{2}]=\{0\}$, for each $(C_{1},C_{2})\in\mathcal{C}_{1}\times_{\mathcal{C}_{1\wedge2}}\mathcal{C}_{2}$. As every normal operator appears in some context, and every operator is a linear combination of normal operators, we conclude that microcausality holds.
\end{proof}

For a net $A$ satisfying the sheaf condition, $h$ needs to be an isomorphism of posets, implying that $\vee$ and $h$ form an adjunction equivalence, which means that the inequalities of the unit and counit of this adjunction are equalities. To be more precise, the sheaf condition implies the equalities
\begin{equation} \label{equ: unit}
C=(C\cap A_{1})\vee(C\cap A_{2});
\end{equation}
\begin{equation} \label{equ: counit}
(C_{1}\vee C_{2})\cap A_{1}=C_{1},\ \ (C_{1}\vee C_{2})\cap A_{2}=C_{2},
\end{equation}
for each $C\in\mathcal{C}_{1\vee2}$ and each $(C_{1},C_{2})\in\mathcal{C}_{1}\times_{\mathcal{C}_{1\wedge2}}\mathcal{C}_{2}$. We already noted that (\ref{equ: unit}) is too restrictive. However, the equality (\ref{equ: counit}), introduced in~\cite{nuiten} as strong locality, does not seem that restrictive at first glance.

\begin{dork}
A net $A:\mathcal{V}(X)\to\mathbf{Sets}$ of operator algebras is called \textbf{strongly local} if it satisfies microcausality and if for any pair $O_{1},O_{2}\in\mathcal{V}(X)$ of spacelike separated opens, equality (\ref{equ: counit}) holds.
\end{dork}

Strong locality states that $h$, seen as a functor of poset-categories, is a coreflector (i.e. it has a left adjoint which is a right inverse). We can describe strong locality as a condition on $H$, instead of $h$. 

\begin{dork}
A geometric morphism $F:\mathcal{E}\to\mathcal{F}$ is called a \textbf{local geometric morphism} is $F_{\ast}$ is full and faithful.
\end{dork}

There are various equivalent ways of stating that a geometric morphism is local (\cite[Theorem C3.6.1]{jh1}). The important point is that for any pair $\mathcal{C}$ and $\mathcal{D}$ of small categories\footnote{Where we assume that $\mathcal{C}$ is Cauchy-complete in the sense that each idempotent morphism splits. As we are concerned with poset categories, this condition holds trivially.}, local geometric morphisms $F:[\mathcal{C},\mathbf{Set}]\to[\mathcal{D},\mathbf{Set}]$ correspond exactly to coreflectors $f:\mathcal{C}\to\mathcal{D}$.

\begin{cor}
A net $A:\mathcal{V}(X)\to\mathbf{Sets}$ of operator algebras is strongly local iff for any pair $O_{1},O_{2}\in\mathcal{V}(X)$ of spacelike separated opens, the geometric morphism $H$ of the descent morphism is local.
\end{cor}

It is tempting to think of strong locality as stating that although $A$ may not be a sheaf, it is infinitesimally close to being one. To make this less sketchy, consider a geometric morphism $F: Sh(Y)\to Sh(X)$ coming from a continuous map $f:Y\to X$, of sober spaces, and assume that $f$ is an infinitesimal thickening. By this we mean that $f$ is a surjection with the property that for each fibre $f^{-1}(x)$ we can pick an element $y_{x}$ such that the only neighbourhood of $y_{x}$ in $f^{-1}(x)$ is  $f^{-1}(x)$ itself, and the assignment $c: x\mapsto y_{x}$ defines a continuous section of $f$ (\cite[C3.6]{jh1}). If this holds, $F$ is a local geometric morphism.

This is relevant to strong locality. If we assume strong locality, and view $h:\mathcal{C}_{1\vee2}\to\mathcal{C}_{1}\times_{\mathcal{C}_{1\wedge2}}\mathcal{C}_{2}$ as an Alexandroff continuous map, then it is an infinitesimal thickening in the sense given above. The continuous section $c$ is given by $\vee$. We find another way of looking at strong locality; the map $h$ is an infinitesimal thickening.

\section{Strong locality and independence conditions} \label{sec: locality}

\subsection{Independence conditions}

The previous section introduced strong locality as a weaker version of the sheaf condition. A net of observable algebras satisfying Einstein causality is strongly local, and any strongly local net must satisfy microcausality. In this section we try to pinpoint strong locality among the various independence conditions used in AQFT. In what follows, we concentrate on pairs $(A,B)$ of unital C*-algebras, instead of whole nets of such algebras. Fo example, think of $A$ and $B$ as operator algebras associated to two spacelike separated regions of spacetime.

\begin{dork}{(\cite{sum})} \label{def: list}
Let $A$ and $B$ be two (not necessarily commutative) unital C*-subalgebras of some larger C*-algebra $\mathfrak{A}$. Then the pair $(A,B)$ satisfies:
\begin{enumerate}
\item \textbf{microcausality} if the elements of $A$ commute with those of $B$, i.e. $[A,B]=\{0\}$;
\item \textbf{extended locality} if it satisfies microcausality and $A\cap B=\mathbb{C}$;
\item \textbf{C*-independence} if it satisfies microcausality and if for every $a\in A$ and $b\in B$, $ab=0$ implies $a=0$ or $b=0$. This condition is usually called the \textbf{Schlieder property};
\item \textbf{C*-independence in the product sense} if it satisfies microcausality and $A\vee B\cong A\otimes B$.
\end{enumerate}
\end{dork}

The locality conditions are sorted in increasing strength. From their definitions we see that C*-independence in the product sense implies C*-independence, and that extended locality implies microcausality. It is not obvious at first sight that C*-independence implies extended locality.

\begin{lem}
C*-independence implies extended locality.
\end{lem}

\begin{proof}
By microcausality, $A\cap B$ is a commutative unital C*-algebra. Hence $A\cap B\cong C(\Sigma)$, where $\Sigma$ is the associated Gelfand spectrum. Under the assumption of microcausality, extended locality is equivalent to the compact Hausdorff space $\Sigma$ being a singleton. We give a contrapositive proof of the lemma. Assume that $x,y\in\Sigma$ are two distinct points. By the Hausdorff property there exist open neighbourhoods $U_{x}$ of $x$ and $U_{y}$ of $y$, such that $U_{x}\cap U_{y}=\emptyset$. A compact Hausdorff space is completely regular, therefore there exist nonzero continuous real-valued functions $f$ and $g$ on $\Sigma$, such that the support of $f$ lies in $U_{x}$ and the support of $g$ lies in $U_{y}$. We found $f,g\in(A\cap B)_{sa}$ such that $f\neq0$, $g\neq0$ and $f\cdot g=0$. This implies that the Schlieder property fails for the pair $(A,B)$.
\end{proof}

The following examples show that none of the conditions of Definition~\ref{def: list} are equivalent.

\begin{exa}
Take $\mathfrak{A}=A_{1}\oplus B_{1}$, with $A_{1}$ and $B_{1}$ C*-algebras. Defining $A=A_{1}\oplus\mathbb{C}$, and $B=\mathbb{C}\oplus B_{1}$, the pair $(A,B)$ satisfies microcausality, but not extended locality, since $A\cap B=\mathbb{C}\oplus\mathbb{C}$.
\end{exa}

\begin{exa} \label{exa: important}
Consider $\mathfrak{A}=C([0,1],\mathbb{C})$, the continuous complex-valued functions on the closed interval $[0,1]$. Define
\begin{equation*}
A=\{f\in\mathfrak{A}\mid f|_{[0,1/2]}\ \text{is constant}\ \};
\end{equation*}
\begin{equation*}
B=\{f\in\mathfrak{A}\mid f|_{[1/2,1]}\ \text{is constant}\ \}.
\end{equation*}
Then the pair $(A,B)$ satisfies extended locality, but the Schlieder property fails.
\end{exa}

\begin{exa}
Let $A$ be a von Neumann factor, and $B=A'$ its commutant. Then the pair $(A,B)$ is C*-independent. But, as shown in~\cite[Corollary 4.6]{effroslance}, it is C*-independent in the product sense iff $A$ is semidiscrete. 
\end{exa}

For commutative C*-algebras, C*-independence and C*-independence in the product sense are equivalent, as shown by the following lemma. The original proof of the lemma is included for the sake of completeness.

\begin{lem}{(\cite{hamhalter} Theorem 11.1.1)}
Let $C$ and $D$ be commutative unital C*-subalgebras of some larger C*-algebra. If $(C,D)$ is C*-independent, then it is C*-independent in the product sense.
\end{lem}

\begin{proof}
Define the $\ast$-homomorphism
\begin{equation*}
\Phi: C\otimes D\to C\vee D,\ \ f\otimes g\mapsto f\cdot g.
\end{equation*}
We will show that $\Phi$ is a $\ast$-isomorphism. Assume that there is an element $h\in Ker\Phi$ that is nonzero and nonnegative. Let $X$ and $Y$ be the Gelfand spectra of $C$ and $D$ respectively. We will use the isomorphism
\begin{equation*}
C\otimes D\to C(X\times Y),\ \ f\otimes g\mapsto u,\ \ u:X\times Y\to\mathbb{C},\ \ u(x,y)=f(x)g(y).
\end{equation*}
Under this isomorphism $h$ can be seen as a nonzero, nonnegative function $h:X\times Y\to\mathbb{R}$. Let $(x,y)\in X\times Y$ be a point such that $h( x,y)>0$. Consider compact neighbourhoods $U_{x}$ of $x$ in $X$, and $U_{y}$ of $y$ in $Y$, such that $h$ restricted to $U_{x}\times U_{y}$ is strictly positive. By compactness of $U_{x}\times U_{y}$, there exists a constant $c>0$, such that $h>c$ on $U_{x}\times U_{y}$. There exists a nonzero nonnegative continuous function $f$ on $X$, that vanishes outside of $U_{x}$. Likewise, there exists a continuous nonzero nonnegative function $g$ that vanishes outside of $U_{y}$. Rescale these two functions such that $0\leq f,g\leq\sqrt{c}$. Define $\bar{h}=f\otimes g$. By construction, $\bar{h}\leq h$, therefore $\Phi(\bar{h})=0$. This means that $f\cdot g=0$, implying that $(C,D)$ is not C*-independent. Contrapositively, if $(C,D)$ is C*-independent, then $Ker\Phi=\{0\}$, implying $C\otimes D\cong C\vee D$. 
\end{proof}

We know from the Proposition~\ref{poe: micro} that microcausality of $(A,B)$ can be described at the level of contexts $\mathcal{C}_{A}$, and $\mathcal{C}_{B}$. Microcausality is equivalent to the claim that the poset-morphism
\begin{equation*}
r:\mathcal{C}_{A\vee B}\to\mathcal{C}_{A}\times_{\mathcal{C}_{A\cap B}}\mathcal{C}_{B}\ \ r(C)=(C\cap A,C\cap B)
\end{equation*}
has a left adjoint. Extended locality is now also easily described at the level of contexts, as it amounts to microcausality combined with the statement that $\mathcal{C}_{A\cap B}$ is a singleton set. Equivalently, the pair $(A,B)$ satisfies extended locality iff it satisfies microcausality and the pullback square
\[ \xymatrix{ \mathcal{C}_{A}\times_{\mathcal{C}_{A\cap B}}\mathcal{C}_{B} \ar[r] \ar[d] & \mathcal{C}_{A} \ar[d]^{(-)\cap B} \\
\mathcal{C}_{B} \ar[r]_{(-)\cap A} & \mathcal{C}_{A\cap B} } \]
is equal to the product $\mathcal{C}_{A}\times\mathcal{C}_{B}= \mathcal{C}_{A}\times_{\mathcal{C}_{A\cap B}}\mathcal{C}_{B}$. We proceed to describe C*-independence at the level of contexts.

\begin{poe}
The pair $(A,B)$ is C*-independent iff 
\begin{equation*}
\forall C\in\mathcal{C}_{A}\ \ \forall D\in\mathcal{C}_{B}\ \ C\vee D\cong C\otimes D.
\end{equation*}
\end{poe}

\begin{proof}
Assume that $(A,B)$ is C*-independent. If $C\in\mathcal{C}_{A}$ and $D\in\mathcal{C}_{B}$, then $(C,D)$ satisfies the Schlieder property, because $(A,B)$ does. The pair $(C,D)$ is C*-independent, which implies C*-independence in the product sense, as we are working with commutative algebras. We conclude that $C\vee D\cong C\otimes D$.

Conversely, assume that $C\vee D\cong C\otimes D$. Then $(C,D)$ is C*-independent, and satisfies the Schlieder property. All normal operators of $A$ and $B$ occur in contexts, so the Schlieder property holds when we restrict $a$ and $b$ to normal operators. Let $a\in A$ and $b\in B$ be arbitrary. Assume that $a\cdot b=0$. Then 
\begin{equation*}
(a^{\ast}a)\cdot(bb^{\ast})= a^{\ast}\cdot(ab)\cdot b^{\ast}=0.
\end{equation*}
By the Schlieder property for normal operators, $a^{\ast}a=0$ or $bb^{\ast}=0$. This implies that $a=0$ or $b=0$, proving the Schlieder property for the pair $(A,B)$. We conclude that $(A,B)$ is C*-independent.
\end{proof}

\subsection{C*-independence and the spectral presheaf} \label{subsub: spectral}

As we will see below, using an equivalent description of C*-independence, this condition resembles a sheaf condition on the topos $[\mathcal{C}^{op},\mathbf{Set}]$. Strictly speaking, it is not really a sheaf condition, because the `covering relation' in question fails to be a basis for a Grothendieck topology. 

By a state on a C*-algebra we mean a normalised positive linear functional on the algebra. As argued in~\cite{sum}, a pair $(A,B)$ is C*-independent iff, for any state $\phi_{1}$ of $A$, and any state $\phi_{2}$ on $B$, there exists a unique state $\phi$ on $A\vee B$, such that
\begin{equation*}
\forall a\in A\ \ \forall b\in B\ \ \phi(a\cdot b)=\phi_{1}(a)\cdot\phi_{2}(b).
\end{equation*}
From the previous proposition it follows that $(A,B)$ is C*-independent iff, for any $C\in\mathcal{C}_{A}$, any state $\phi_{1}$ on $C$, any $D\in\mathcal{C}_{B}$, and any state $\phi_{2}$ on $D$, there exists a unique state $\phi$ on $C\vee D$ such that $\phi(ab)=\phi_{1}(a)\phi_{2}(b)$.\\

Let $\Sigma_{C}$ denote the Gelfand spectrum of the context $C$, and let $\text{PV}(\Sigma_{C})$ denote the set of probability valuations on $\Sigma_{C}$. See e.g.~\cite{hls} for a precise definition of a probability valuations $\mu:\mathcal{O}\Sigma_{C}\to[0,1]_{l}$. If $C_{1}\subseteq C_{2}$, let $\rho_{C_{2}C_{1}}:\Sigma_{C_{2}}\to\Sigma_{C_{1}}$ denote the Gelfand dual of the inclusion map $C_{1}\hookrightarrow C_{2}$. Define the function
\begin{equation*}
\text{PV}(i_{C_{1}C_{2}}):\text{PV}(\Sigma_{C_{2}})\to\text{PV}(\Sigma_{C_{1}})\ \ \mu\mapsto \mu\circ\rho^{-1}_{C_{2}C_{1}}.
\end{equation*}
For any unital C*-algebra $A$, this assignment defines a presheaf 
\begin{equation*}
\text{PV}(\underline{\Sigma}):\mathcal{C}_{A}^{op}\to\underline{Set}.
\end{equation*}
For an element $\lambda\in\Sigma_{C_{2}}$, let $\delta_{\lambda}\in\text{PV}(\Sigma_{C_{2}})$ denote the point valuation satisfying $\delta_{\lambda}(U)=1$ iff $\lambda\in U$ and $\delta_{\lambda}(U)=0$ otherwise. By definition,
\begin{equation*}
\text{PV}(i_{C_{2}C_{1}})(\delta_{\lambda})=\delta_{\lambda}\circ\rho^{-1}_{C_{2}C_{1}}=\delta_{\rho_{C_{2}C_{1}}(\lambda)}.
\end{equation*}
This allows us to see the spectral presheaf $\underline{\Sigma}$ as a subobject of the presheaf $\text{PV}(\underline{\Sigma})$. The spectral presheaf, which is a central object in the work of Isham and collaborators~\cite{butish1,di}, is the presheaf assigning to each context $C$, its Gelfand spectrum $\Sigma{C}$, and assigning to the inclusion $C_{1}\subseteq C_{2}$, the (continuous) restriction map $\rho_{C_{2}C_{1}}:\Sigma_{C_{2}}\to\Sigma_{C_{1}}$.

\begin{rem}
Given a locale $\underline{L}$ in the presheaf topos $[\mathcal{C}^{op},\mathbf{Set}]$, we can assign to it the locale $\underline{\text{PV}}(\underline{L})$ of internal probability valuations on it. This assignment is part of an endofunctor
\begin{equation*}
\underline{\text{PV}}:\underline{Loc}_{[\mathcal{C}^{op},\mathbf{Set}]}\to\underline{Loc}_{[\mathcal{C}^{op},\mathbf{Set}]},
\end{equation*}
which is, in turn, part of a monad~\cite{vic7}. The inclusion $\underline{\Sigma}\subseteq\text{PV}(\underline{\Sigma})$ resembles of the unit $\eta: I\to\underline{\text{PV}}$ of this monad. Clearly, the inclusion cannot be completely identified as $\eta_{\underline{\Sigma}}$, as we do not view $\underline{\Sigma}$ as an internal locale here, and neither do we consider \textbf{internal} probability valuations in the definition of $PV(\uS)$.
\end{rem}

\begin{rem}
For finite-dimensional C*-algebras $A$, the presheaf $\text{PV}(\underline{\Sigma})$ (or rather the restriction of it to a finite subset of contexts) is easily identified with the quantum-mechanical realisation of the presheaf of $R$-distributions on the event sheaf $D_{R}(\mathcal{E})$, used in~\cite{ab}. Here $R$ is taken to be the non-negative real numbers.
\end{rem}

Using the Riesz--Markov theorem, C*-independence of the pair $(A,B)$ can be translated to a condition on $\text{PV}(\underline{\Sigma}):\mathcal{C}_{A\vee B}^{op}\to\mathbf{Set}$, which resembles a sheaf condition. For any $C\in\mathcal{C}_{A}\subseteq\mathcal{C}_{A\vee B}$, any $D\in\mathcal{C}_{B}\subseteq\mathcal{C}_{A\vee B}$, for each pair of probability valuations $\mu_{1}\in\text{PV}(\Sigma_{C})$, and $\mu_{2}\in\text{PV}(\Sigma_{D})$, there exists a unique $\mu\in\text{PV}(\Sigma_{C\vee D})$ such that
\begin{equation*}
\mu|_{C}:=\text{PV}(i_{C,C\vee D})(\mu)=\mu_{1},\ \ \ \mu|_{D}:=\text{PV}(i_{D,C\vee D})(\mu)=\mu_{2},
\end{equation*}
and, in addition, $\mu$ can be identified as the product probability valuation $\mu=\mu_{1}\times\mu_{2}$ (see~\cite{vic7} for a discussion of product valuations).

If we think of the pair $(C,D)$ as covering $C\vee D$, this resembles a sheaf condition on $\text{PV}(\underline{\Sigma})$. Let us make this precise. Define a `covering relation' $\drie$, where for $E\in\mathcal{C}_{A\vee B}$, and $U\subseteq\mathcal{C}_{A\vee B}$, $E\drie U$ means that $E$ is covered by $U$. For each $E\in\mathcal{C}_{A\vee B}$ define $E\drie\{E\}$, and, if we can write $E=C\vee D$ with $C\in\mathcal{C}_{A}$, and $D\in\mathcal{C}_{B}$, then $E\drie\{C,D\}$ as well. However, this relation $\drie$ does not satisfy the necessary conditions for a basis for a Grothendieck topology on $\mathcal{C}_{A\vee B}$, in the sense of~\cite[III.2, Def. 2]{mm}. The obstruction is the stability axiom, which in our setting requires that for any $E\in\mathcal{C}_{A\vee B}$ such that $E\subseteq C\vee D$ for some $C\in\mathcal{C}_{A}$ and $D\in\mathcal{C}_{B}$, one has
\begin{equation*}
E=(E\wedge C)\vee(E\wedge D).
\end{equation*}
Looking at Example~\ref{exa: counter} this condition does not hold. Indeed, the relation $\drie$ does not define a basis for a Grothendieck topology for the same reason that the full sheaf condition of Definition~\ref{def: sheafcond} does not hold.

\subsection{Pinpointing strong locality}

Next, we try to relate the previous locality conditions to the notion of strong locality. Recall:

\begin{dork}
Let $A$ and $B$ be two not necessarily commutative unital C*-subalgebras of some larger C*-algebra $\mathfrak{A}$. Then the pair $(A,B)$ is called \textbf{strongly local} if it satisfies microcausality, and 
\begin{equation} \label{equ: strongloc}
\forall C\in\mathcal{C}_{A}\ \  \forall D\in\mathcal{C}_{B}\ \ (C\vee D)\cap A=C\ \text{and}\ (C\vee D)\cap B=D.
\end{equation}
\end{dork}

We wish to see where strong locality stands in the list of Definition~\ref{def: list}. According to the following lemma it implies extended locality. 

\begin{lem}
If the pair $(A,B)$ is strongly local, then it satisfies extended locality.
\end{lem}

\begin{proof}
The proof is this lemma is contrapositive. Assume that $(A,B)$ satisfies microcausality, but not extended locality. Then $A\cap B\neq\mathbb{C}$. Take $C=\mathbb{C}\in\mathcal{C}_{A}$ and $D=A\cap B\in\mathcal{C}_{B}$. Then
\begin{equation*}
(C\vee D)\cap A=(\mathbb{C}\vee(A\cap B))\cap A=(A\cap B)\cap A=A\cap B=D\neq C.
\end{equation*}
Hence, the pair $(A,B)$ is not strongly local.
\end{proof}

Note that if $(A,B)$ are C*-independent in the product sense, then the pair satisfies strong locality. We use this observation to prove that C*-independence implies strong locality.

\begin{lem}
If the pair $(A,B)$ is C*-independent, then $(A,B)$ satisfies strong locality.
\end{lem}

\begin{proof}
Let $C\in\mathcal{C}_{A}$ and $D\in\mathcal{C}_{B}$. Define $E=(C\vee D)\cap A$. Then $E$ is a commutative C*-algebra containing $C$ and is contained in $A$. Likewise, $F=(C\vee D)\cap B$ is a commutative C*-algebra that contains $D$, and is contained in $B$. As the pair $(A,B)$ is C*-independent, so is the pair $(E,F)$. Since $E$ and $F$ are commutative, the pair $(E,F)$ is also C*-independent in the product sense. Hence, it is strongly local. Note that
\begin{equation*}
C=(C\vee D)\cap E=(C\vee D)\cap(C\vee D)\cap A=(C\vee D)\cap A,
\end{equation*}
and similarly $D=(C\vee D)\cap B$. We conclude that the pair $(A,B)$ is strongly local as well.
\end{proof}

C*-independence implies strong locality, which in turn implies extended locality. Example~\ref{exa: important} can be used to show that strong locality and C*-independence are not equivalent.

\begin{lem}
The pair $(A,B)$ from Example~\ref{exa: important} satisfies strong locality.
\end{lem}

\begin{proof}
For this proof, we use the following observations. Let $E$ be a unital commutative C*-subalgebra of $C([0,1])$. This algebra $E$ defines a partition of the interval $[0,1]$ into closed subsets by means of the equivalence relation
\begin{equation*}
x\sim_{E} y\ \ \text{iff}\ \ \forall f\in E\ \ f(x)=f(y).
\end{equation*}
The algebra $E$ consists of those $f\in C([0,1])$ for which $f$ is constant on each of the equivalence classes $[x]_{E}$. Conversely, if we have a partition of $[0,1]$ into closed subsets $[x]$, then the set of $f\in C([0,1])$ such that $f$ is constant on each equivalence class $[x]$ defines a C*-subalgebra of $C([0,1])$.

Take $C\in\mathcal{C}_{A}$ and $D\in\mathcal{C}_{D}$. These correspond to partitions $[x]_{C}$ and $[x]_{D}$ of $[0,1]$. Note that by definition of $A$ and $B$, $[1/2]_{C}$ contains the interval $[0,1/2]$, and $[1/2]_{D}$ contains $[1/2,1]$. Define the finer partition consisting of classes $[x]=[x]_{D}\cap[x]_{C}$. Note that if $x\notin[1/2]_{C}$, then $[x]=[x]_{C}$, and if $x\notin[1/2]_{D}$, then $[x]=[x]_{D}$.

Let $E$ be the C*-algebra consisting of those $h\in C([0,1])$ that are constant on each class $[x]$. Define
\begin{equation*}
f:[0,1]\to\mathbb{C},\ \ x\mapsto\left\{
\begin{array}{ll}
h(1/2) & \text{if } x\in[1/2]_{C}\\
h(x) & \text{if } x\notin [1/2]_{C}
\end{array}. \right.
\end{equation*}
Note that $f\in C$. In addition, define
\begin{equation*}
g:[0,1]\to\mathbb{C},\ \ x\mapsto\left\{
\begin{array}{ll}
h(x) & \text{if } x\notin[1/2]_{D}\\
h(1/2) & \text{if } x\in [1/2]_{D}
\end{array}. \right.
\end{equation*}
It follows that $g\in D$. Note that $h=f+g-h(1/2)$. As a consequence $E\subseteq C\vee D$. It is straightforward to verify the converse $C\vee D\subseteq E$. We conclude that $C\vee D= E$. To demonstrate strong locality we need to show that $E\cap A\subseteq C$ and $E\cap B\subseteq D$. Let $h\in E\cap A$. then, for each $x\in[0,1]$, $h$ is constant on $[x]_{D}\cap [x]_{C}$. Let $[x]_{C}\neq[1/2]_{C}$. Then $[x]_{C}\cap[x]_{D}=[x]_{C}$. We conclude that $h$ is constant on each $[x]_{C}$, except possibly $[1/2]_{C}$. As $h\in A$, $h$ is constant on $[0,1/2]$, as well as on $[1/2]_{C}\cap[1/2]_{D}$. Note that
\begin{align*}
[0,1/2]\cup([1/2]_{D}\cap[1/2]_{C})&=([0,1/2]\cup[1/2]_{C})\cap([0,1/2]\cup[1/2]_{D})\\
&=[1/2]_{C}\cap[0,1]=[1/2]_{C}.
\end{align*}
We conclude that $h\in C$. This shows that $(C\vee D)\cap A= C$. The equality $(C\vee D)\cap B= D$ can be proven in the same way.
\end{proof}

It is an open question whether strong locality is equivalent to extended locality, or is a new conditions.

\section{C*-Algebraic version} \label{sec: pullback}

In Section~\ref{sec: sheaf} a sheaf condition was derived using the category $\mathbf{RingSp}$. This meant reducing internal unital commutative C*-algebras to rings. The aim of this section is to formulate the sheaf condition in the category $\mathbf{ucCSp}$ of spaces $X$ and unital commutative C*-algebras internal to $Sh(X)$. The sheaf condition obtained in this way is equivalent to the unit law (\ref{equ: unit}) together with C*-independence.  We start by considering the following C*-algebraic counterpart to $\mathbf{RingTopos}$, which was introduced in~\cite{nuiten}.

\begin{dork}
The category $\mathbf{ucCTopos}$ consists of the following:
\begin{itemize}
\item An \textit{Object} is a pair $(\mathcal{E},\underline{A})$, where $\mathcal{E}$ is a topos and $\underline{A}$ is a unital commutative C*-algebra internal to the topos $\mathcal{E}$.
\item An \textit{arrow} $(G,\underline{g}):(\mathcal{E},\underline{A})\to(\mathcal{F},\underline{B})$ is given by a geometric morphism $G:\mathcal{E}\to\mathcal{F}$ and a $\ast$-homomorphism $\underline{g}:G^{\ast}\underline{B}\to\underline{A}$ in $\mathcal{F}$.
\item Composition of arrows is defined by $(G,\underline{g})\circ(F,\underline{f})=(G\circ F,\underline{f}\circ F^{\ast}\underline{g})$.
\end{itemize}
\end{dork}

As C*-algebras cannot be completely captured by a geometric theory, we cannot guarantee that $G^{\ast}\underline{B}$ is a C*-algebra internal to $\mathcal{F}$. We do know that it is a semi-normed $\ast$-algebra over the complex rationals, where the semi-norm has the C*-property. If by a $\ast$-homomorphism we mean a function that preserves all algebraic structure, then this definition makes sense. But we can question if it is the right definition.

For the semi-normed $\ast$-algebra $G^{\ast}\underline{B}$ we can consider the Cauchy-completion and obtain a C*-algebra. Just as in $\mathbf{Set}$, the algebra $G^{\ast}\underline{B}$ is everywhere norm-dense in its completion. In order to extend a $\ast$-homomorphism (in the algebraic sense given above) $\underline{g}:G^{\ast}\underline{B}\to\underline{A}$ to the completion of $G^{\ast}\underline{B}$ in a continuous way, we require $\underline{g}$ to satisfy
\begin{equation} \label{equ: conti}
\forall b\in G^{\ast}\ \underline{B}\ \forall q\in\mathbb{Q}^{+}\ \ (b,q)\in \uN_{G^{\ast}\underline{B}}\rightarrow (\underline{g}(b),q)\in\uN_{\underline{A}},
\end{equation}
which simply states that $\Vert\underline{g}(b)\Vert\leq\Vert b\Vert$. In the topos $\mathbf{Set}$, this is a necessary condition as every $\ast$-homomorphism between C*-algebras is norm decreasing. We demand that each $\ast$-homomorphism in our category is continuous in this sense. Note that (\ref{equ: conti}) is an expression within geometric logic. If $\underline{f}$ is a continuous $\ast$-homomorphism in a topos $\mathcal{E}$ and $F:\mathcal{F}\to\mathcal{E}$ is a geometric morphism, then $F^{\ast}\underline{f}$ is a continuous $\ast$-homomorphism. This is important, as it entails that continuous $\ast$-homomorphisms are closed under the composition of arrows in $\mathbf{ucCTopos}$. 

\begin{rem}
Let $\uA$ and $\underline{B}$ be semi-normed $\qui$-algebras. If we see $\uA$ and $\underline{B}$ as models of the geometric theory of such algebras, then the structure homomorphisms of these models are precisely the continuous $\ast$-homomorphisms (see e.g.~\cite[Definition 1.20]{vic}).
\end{rem}

In formulating the C*-algebraic version of Nuiten's sheaf condition we use the following subcategory of $\mathbf{ucCTopos}$.

\begin{dork}
The category $\mathbf{ucCSp}$ consists of the following:
\begin{itemize}
\item An \textit{Object} is a pair $(X,\underline{A})$, where $X$ is a topological space and $\underline{A}$ is a unital commutative C*-algebra internal to the topos $Sh(X)$.
\item An \textit{arrow} $(g,\underline{g}):(X,\underline{A})\to(Y,\underline{B})$ is given by a continuous map $g:X\to Y$ and a continuous, unit-preserving $\ast$-homomorphism $\underline{g}:G^{\ast}\underline{B}\to\underline{A}$ in $Sh(Y)$, where $G$ is the geometric morphism induced by $g$.
\item Composition of arrows is defined by $(G,\underline{g})\circ(F,\underline{f})=(G\circ F,\underline{f}\circ F^{\ast}\underline{g})$.
\end{itemize}
\end{dork}

Let $\mathcal{O}\to A(\mathcal{O})$ be an isotonic net of operator algebras, assumed to be unital C*-algebras. Such a net defines a contravariant functor $A:\mathcal{V}(X)^{op}\to\mathbf{ucCSp}$. Let $O_{1}$, and $O_{2}$ be two spacelike separated opens. If we think of $O_{1}\vee O_{2}$ as being covered by $O_{1}$ and $O_{2}$, we want to know under what conditions $A$ is a sheaf, just like in Section~\ref{sec: sheaf}. Recall that the sheaf condition was formulated using the descent morphism. This descent morphism used the ringed topos of matching families, which was defined as a pullback. In the C*-algebraic version, does the corresponding pullback even exist? Proving that such pullbacks exist will be our first step in describing the sheaf condition.

Assume for convenience that the net satisfies extended locality. Using the notation of Section~\ref{sec: sheaf}, consider the poset $\mathcal{C}_{1}\times_{\mathcal{C}_{1\wedge2}}\mathcal{C}_{2}$, which, under the assumption of extended locality, simplifies to $\mathcal{C}_{1}\times\mathcal{C}_{2}$. As before, we use the topos $\mathcal{F}=[\mathcal{C}_{1}\times\mathcal{C}_{2},\mathbf{Set}]$. In this topos, define the internal unital commutative C*-algebra
\begin{equation}
\underline{A}_{1}\otimes\underline{A}_{2}:\mathcal{C}_{1}\times\mathcal{C}_{2}\to\mathbf{Set}\ \ (\underline{A}_{1}\otimes\underline{A}_{2})(C_{1},C_{2})=C_{1}\otimes C_{2},
\end{equation}
where $C_{1}\otimes C_{2}$ is the C*-algebraic tensor product\footnote{Since commutative C*-algebras are nuclear, there is a unique C*-algebraic tensor product completing the algebraic tensor product $C_{1}\odot C_{2}$.}. If $(D_{1},D_{2})\leq(C_{1},C_{2})$, The corresponding $\ast$-homomorphism is the inclusion $D_{1}\otimes D_{2}\hookrightarrow C_{1}\otimes C_{2}$. The pair $(\mathcal{F},\underline{A}_{1}\otimes\underline{A}_{2})$  comes equipped with projection morphisms
\begin{equation*}
(P_{i},\underline{p}_{i}):([\mathcal{C}_{1}\times\mathcal{C}_{2},\mathbf{Set}],\underline{A}_{1}\otimes\underline{A}_{2})\to([\mathcal{C}_{i},\mathbf{Set}],\underline{A}_{i}),
\end{equation*}
where the geometric morphism $P_{i}$ is induced by the projection $\pi_{i}:\mathcal{C}_{1}\times\mathcal{C}_{2}\to\mathcal{C}_{i}$, and 
\begin{equation*}
\underline{p}_{i}:P^{\ast}_{i}\underline{A}_{i}\to\underline{A}_{1}\otimes\underline{A}_{2},
\end{equation*}
\begin{equation*}
(\underline{p}_{1})_{(C_{1},C_{2})}:C_{1}\to C_{1}\otimes C_{2}\ \ a\mapsto a\otimes 1,
\end{equation*}
\begin{equation*}
(\underline{p}_{2})_{(C_{1},C_{2})}:C_{2}\to C_{1}\otimes C_{2}\ \ b\mapsto 1\otimes b.
\end{equation*}

\begin{tut} \label{thm: prod}
The diagram
\[ \xymatrix{
(\mathcal{C}_{1},\underline{A}_{1}) & & (\mathcal{C}_{1}\times\mathcal{C}_{2},\underline{A}_{1}\otimes\underline{A}_{2}) \ar[ll]_{(P_{1},\underline{p}_{1})} \ar[rr]^{(P_{2},\underline{p}_{2})} & & (\mathcal{C}_{2},\underline{A}_{2}), } \]
is a product in $\mathbf{ucCSp}$.
\end{tut}

Before we start with the proof, note the similarity with the pullback in Section~\ref{sec: sheaf}. By extended locality $(\mathcal{T}_{1\wedge2},\underline{A}_{1\wedge2})$, can be identified with the pair $(\mathbf{Set},\mathbb{C})$, which is the terminal object of $\mathbf{ucCSp}$. It is by extended locality that the pullback coincides with the product. 

\begin{proof}

The proof is presented in six steps.
\begin{enumerate}

\item We start by showing that
\[ \xymatrix{
P^{\ast}_{1}\underline{A}_{1} & \underline{A}_{1}\otimes\underline{A}_{2} \ar[r]^{\underline{p}_{1}} \ar[l]_{\underline{p}_{2}} & P^{\ast}_{2}\underline{A}_{2},} \]
is a coproduct of unital commutative C*-algebras in $[\mathcal{C}_{1}\times\mathcal{C}_{2},\mathbf{Set}]$. An exercise in presheaf semantics, or~\cite{svw}, shows that $\underline{B}$ is a unital commutative C*-algebra in $\mathcal{F}$ iff each $\underline{B}(C_{1},C_{2})$ is a unital commutative C*-algebra in $\mathbf{Set}$, and each restriction map is a unit-preserving $\ast$-homomorphism. Note that for each $(C_{1},C_{2})\in\mathcal{C}_{1}\times\mathcal{C}_{2}$, the diagram above gives
\[ \xymatrix{C_{1} & C_{1}\otimes C_{2} \ar[r]^{1\otimes-} \ar[l]_{-\otimes1} & C_{1} }, \]
which is a coproduct of unital commutative C*-algebras in $\mathbf{Set}$. Suppose we have internal $\ast$-homomorphisms (which in this case are automatically norm-continuous) $\underline{f}_{i}:P^{\ast}_{i}\underline{A}_{i}\to\underline{B}$. This provides us with $\ast$-homomorphisms 
\begin{equation*}
\forall (C_{1},C_{2})\in\mathcal{C}_{1}\times\mathcal{C}_{2}\ \ (\underline{f}_{i})_{(C_{1},C_{2})}:C_{i}\to\underline{B}(C_{1},C_{2}).
\end{equation*}
Using the universal property of the stagewise tensor products, we obtain unique maps
\begin{equation*} 
\underline{h}_{(C_{1},C_{2})}:C_{1}\otimes C_{2}\to\underline{B}(C_{1},C_{2}),
\end{equation*}
\begin{equation} \label{equ: tensie}
\underline{h}_{(C_{1},C_{2})}(a\otimes b)=(\underline{f}_{1})_{(C_{1},C_{2})}(a)(\underline{f}_{2})_{(C_{1},C_{2})}(b).
\end{equation}
Do these local $\ast$-homomorphisms $\underline{h}_{(C_{1},C_{2})}$ piece together to a single natural transformation? Let 
\begin{equation*}
\rho:\underline{B}(D_{1}, D_{2})\to\underline{B}(C_{1},C_{2}),\ \ (D_{1},D_{2})\leq(C_{1},C_{2}),
\end{equation*}
denote the map of the functor $\underline{B}$, corresponding to the given inequality of $\mathcal{C}_{1}\times\mathcal{C}_{2}$. Using naturality of $\underline{f}_{1}$, $\underline{f}_{2}$, and (\ref{equ: tensie}), we find that $\rho\circ\underline{h}_{(D_{1},D_{2})}$ and $\underline{h}_{(C_{1},C_{2})}$ agree on the algebraic tensor product $D_{1}\odot D_{2}$. By continuity, they agree on $D_{1}\otimes D_{2}$, showing naturality. Hence, we found an internal $\ast$-homomorphism $\underline{h}:\underline{A}_{1}\otimes\underline{A}_{2}\to\underline{B}$ with the desired properties. Note that $\underline{h}$ is unique, as each component $\underline{h}_{(C_{1},C_{2})}$ is unique.

\item In this step we concentrate on the algebraic part of $\underline{A}_{1}\otimes\underline{A}_{2}$. Since we have established that $\underline{A}_{1}\otimes\underline{A}_{2}$ is a coproduct of unital commutative C*-algebras, we ask how this coproduct behaves under the action of an inverse image functor coming from a geometric morphism. In the setting of rings this was straightforward, as we were dealing with an essentially algebraic theory. The theory of C*-algebras is not geometric, let alone essentially algebraic. Instead of working with the whole C*-algebra, forget about the norm for a moment. We consider $\ast$-algebras over $\qui$. Looking at $P^{\ast}_{1}\underline{A}_{1}$ and $P^{\ast}_{2}\underline{A}_{2}$ as such algebras, we can find their coproduct $\underline{A}_{1}\odot\underline{A}_{2}$ in $[\mathcal{C}_{1}\times\mathcal{C}_{2},\mathbf{Set}]$. As the notation suggests, this coproduct is also computed context-wise, i.e. $(\underline{A}_{1}\odot\underline{A}_{2})(C_{1},C_{2})=C_{1}\odot C_{2}$. Furthermore, if $H:\mathcal{E}\to[\mathcal{C}_{1}\times\mathcal{C}_{2},\mathbf{Set}]$ is a geometric morphism, then $H^{\ast}(\underline{A}_{1}\odot\underline{A}_{2})$ can be identified with the coproduct of $H^{\ast}P^{\ast}_{1}\underline{A}_{1}$ and $H^{\ast}P^{\ast}_{2}\underline{A}_{2}$. 

The functor $H^{\ast}$ preserves coproducts of $\qui$-algebras for the same reasons it preserves coproducts of rings. If $F:\mathcal{E}\to\mathcal{F}$ is any geometric morphism, and $\uA$ is a $\underline{R}$-algebra in $\mathcal{F}$ for some commutative ring, then $F^{\ast}\uA$ is an $F^{\ast}\underline{R}$-algebra in $\mathcal{E}$. If $\underline{B}$ is a $F^{\ast}\underline{R}$-algebra in $\mathcal{E}$ with action $\mu_{\underline{B}}:F^{\ast}\underline{R}\times\underline{B}\to\underline{B}$, then $F_{\ast}\underline{B}$ is a $\underline{R}$-algebra in $\mathcal{F}$ with action
\begin{equation*}
\mu_{F_{\ast}\underline{B}}:\underline{R}\times F_{\ast}\underline{B}\to F_{\ast}\underline{B},\ \ \mu_{F_{\ast}\underline{B}}=F_{\ast}(\mu_{\underline{B}})\circ(\eta_{\underline{R}}\times F_{\ast}\underline{B}),
\end{equation*}
where $\eta_{\underline{R}}:\underline{R}\to F_{\ast}F^{\ast}\underline{R}$ is the unit of the adjunction. In this way, the $F^{\ast}\underline{R}$-algebra homomorphisms $F^{\ast}\uA\to\underline{B}$, in $\mathcal{E}$ correspond, by the adjunction, to $\underline{R}$-algebra homomorphisms $\uA\to F_{\ast}\underline{B}$ in $\mathcal{F}$.

Returning to case at hand; if we are given $\ast$-preserving algebra morphisms $\underline{f}_{i}: H^{\ast}P^{\ast}_{i}\underline{A}_{i}\to\underline{B}$, then there exists a unique $\ast$-preserving algebra morphism $\underline{h}: H^{\ast}(\underline{A}_{1}\odot\underline{A}_{2})\to\underline{B}$ such that $\underline{h}\circ H^{\ast}\underline{p}_{i}=\underline{f}_{i}$. 

\item The previous step showed that $H^{\ast}(\uA_{1}\odot\uA_{2})$, was the coproduct of $H^{\ast}\uA_{1}$ and $H^{\ast}\uA_{2}$, as $\qui$-algebras in $\mathcal{E}$. As such, for $\ast$-preserving $\qui$-algebra homomorphisms $\underline{f}_{i}: H^{\ast}P_{i}^{\ast}\uA_{i}\to\underline{B}$, there exists a corresponding unique $\ast$-preserving $\qui$-algebra homomorphisms $\underline{h}: H^{\ast}(\uA_{1}\odot\uA_{2})\to\underline{B}$. 

From this point onwards we include the norms back into the discussion. We consider $\uA_{1}\odot\uA_{2}$ normed, using the restriction of the norm $\uN$ on $\underline{A}_{1}\otimes\underline{A}_{2}$, to the subobject $\underline{A}_{1}\odot\uA_{2}$. Equipped with this norm, the last three steps are devoted to showing that $\underline{h}$ is norm-continuous whenever both $\underline{f}_{i}$ are norm-continuous. 

Note that with respect to the norm $\uN$ on $\underline{A}_{1}\otimes\underline{A}_{2}$, the subset $\underline{A}_{1}\odot\underline{A}_{2}$ is everywhere dense in the sense 
\begin{equation*}
\forall x\in\underline{A}_{1}\otimes\underline{A}_{2}\ \forall n\in\mathbb{N}\ \ (\exists y\in\underline{A}_{1}\odot\underline{A}_{2}\ \ (x-y,1/n)\in\uN).
\end{equation*}
This is a suitable axiom for a geometric theory, and this remark is relevant in two ways. First, the axiom holds internally for $\underline{A}_{1}\otimes\underline{A}_{2}$, because it holds at each stage $(C_{1},C_{2})$ (see e.g.~\cite[Cor. D1.2.14(i)]{jh1}). Second, it also holds for $H^{\ast}(\underline{A}_{1}\otimes\underline{A}_{2})$ and $H^{\ast}(\underline{A}_{1}\odot\underline{A}_{2})$, relative to the semi-norm $H^{\ast}\uN$. This implies that the elements of $H^{\ast}(\underline{A}_{1}\otimes\underline{A}_{2})$ can be seen as Cauchy-approximations of $H^{\ast}(\underline{A}_{1}\odot\underline{A}_{2})$, relative to $H^{\ast}\uN$. This, in turn, implies that we can extend any norm-continuous $\ast$-homomorphism $\underline{h}: H^{\ast}(\underline{A}_{1}\odot\underline{A}_{2})\to\underline{B}$ to $H^{\ast}(\underline{A}_{1}\otimes\underline{A}_{2})$.

\item For any pair $\underline{f}_{i}: H^{\ast}P^{\ast}_{i}\underline{A}_{i}\to\underline{B}$ of continuous $\ast$-homomorphisms, we show that the $\ast$-preserving algebra morphism $\underline{h}: H^{\ast}(\underline{A}_{1}\odot\underline{A}_{2})\to\underline{B}$ satisfying $\underline{h}\circ H^{\ast}\underline{p}_{i}=\underline{f}_{i}$, is norm-decreasing for elements of the form $H^{\ast}\underline{p}_{1}(a)=:a\otimes 1$ and $H^{\ast}\underline{p}_{2}(b)=:1\otimes b$.

Note that the norm on $\uA_{1}\odot\uA_{2}$ satisfies the following (geometric) condition, expressing that $\Vert a\otimes1\Vert=\Vert a\Vert$, and $\Vert1\otimes b\Vert=\Vert b\Vert$:
\begin{equation*}
\forall a\in P^{\ast}_{i}\underline{A}_{i}\ \ \forall q\in\mathbb{Q}^{+}\ \ (a,q)\in\uN_{P^{\ast}_{i}\underline{A}_{i}}\leftrightarrow(\underline{p}_{i}(a),q)\in\uN_{\underline{A}_{1}\odot\underline{A}_{2}}.
\end{equation*}
Fitting within geometric logic this identity is preserved by $H^{\ast}$.  Consequently,
\begin{equation*}
(H^{\ast}\underline{p}_{i}(a),q)\in\uN_{H^{\ast}(\underline{A}_{1}\odot\underline{A}_{2})}\rightarrow(a,q)\in\uN_{H^{\ast}P_{i}^{\ast}\underline{A}_{i}}\rightarrow(\underline{f}_{i}(a),q)\in\uN_{\underline{B}},
\end{equation*}
where we used that by definition $\uN_{H^{\ast}P^{\ast}_{i}\underline{A}_{i}}=H^{\ast}\uN_{P^{\ast}_{i}\underline{A}_{i}}$. This argument shows that $\underline{h}$ is (semi-)norm-continuous on elements of the form $\underline{p}_{1}(a)=a\otimes1$ and $\underline{p}_{2}(b)=1\otimes b$.

\item Next, we show that $\underline{h}$ is norm-decreasing for simple tensors $a\otimes b=H^{\ast}\underline{p}_{1}(a)\cdot H^{\ast}\underline{p}_{2}(b)$. We pursue the same strategy as in the previous step. Take a property of the tensor product in $\mathbf{Set}$. Describe it geometrically. Then it holds internally to $\mathcal{F}$, and is preserved by $H^{\ast}$. In $\mathbf{Set}$ the norm on the tensor product is a cross-norm, which means that it satisfies $\Vert a\otimes b\Vert=\Vert a\Vert\cdot\Vert b\Vert$. We can reformulate this as as
\begin{align*}
\forall a\in P^{\ast}_{1}\underline{A}_{1}\ \ &\forall b\in P^{\ast}_{2}\underline{A}_{2} \ \ \forall q\in\mathbb{Q}^{+}\ (\underline{p}_{1}(a)\cdot\underline{p}_{2}(b),q)\in\uN_{\underline{A}_{1}\odot\underline{A}_{2}} \leftrightarrow \\
& (\exists p_{1},p_{2}\in\mathbb{Q}^{+} ((a,p_{1})\in\uN_{P^{\ast}_{1}\underline{A}_{1}})\wedge((b,p_{2})\in\uN_{P^{\ast}_{2}\underline{A}_{2}})\wedge(p_{1}\cdot p_{2}<q)).
\end{align*}
This condition states that for any positive rational $q>0$, $\Vert a\otimes1\cdot1\otimes b\Vert<q$ iff there exist rational numbers $p_{1},p_{2}>0$, satisfying
\begin{equation*}
\Vert a\otimes 1\Vert\cdot\Vert1\otimes b\Vert<p_{1}\cdot p_{2}<q.
\end{equation*}
Note that the condition is preserved by $H^{\ast}$, as the axiom can be expressed using geometric logic. 

We are now ready to prove norm-continuity for simple tensors. Let $q>\Vert H^{\ast}\underline{p}_{1}(a)\cdot H^{\ast}\underline{p}_{2}(b)\Vert$. There exist $p_{1}, p_{2}$ as above. If $p_{1}>\Vert a\Vert$, then $p_{1}>\Vert\underline{f}_{1}(a)\Vert$. Likewise, $p_{2}>\Vert\underline{f}_{2}(b)\Vert$. Using submultiplicativity of the norm, we conclude that
\begin{equation*}
q>p_{1}\cdot p_{2}>\Vert\underline{f}_{1}(a)\Vert\cdot\Vert\underline{f}_{2}(b)\Vert\geq\Vert\underline{f}_{1}(a)\cdot\underline{f}_{2}(b)\Vert=\Vert\underline{h}(H^{\ast}\underline{p}_{1}(a)\cdot H^{\ast}\underline{p}_{2}(b))\Vert.
\end{equation*}
This proves norm-continuity of $\underline{h}$ on the simple tensors. 

\item For the last step in demonstrating norm-continuity, we consider linear combinations of simple tensors, i.e. arbitrary elements of the algebraic tensor product. The proof relies on the fact that the norm which we defined on $\underline{A}_{1}\odot\underline{A}_{2}$ is the projective cross-norm. In the topos $\mathbf{Set}$, for two unital commutative C*-algebras $A$ and $B$, a $q\in\mathbb{Q}^{+}$, and an element $x\in A\odot B$, we have $q>\Vert x\Vert$ iff
\begin{equation*}
\exists n\in\mathbb{N}\ \ \exists a_{1},...,a_{n}\in A\ \ \exists b_{1},...,b_{n}\in B\ \ \left(x=\sum_{i=1}^{n}a_{i}\otimes b_{i}\right)\wedge\left(q>\sum_{i=1}^{n}\Vert a_{i}\Vert\cdot\Vert b_{i}\Vert\right).
\end{equation*}
Noting that using finite subsets or finite lists falls within the domain of geometric logic, this property can be described geometrically. Therefore, it can be applied to $H^{\ast}(\underline{A}_{1}\odot\underline{A}_{2})$. Using both the triangle inequality and the submultiplicativity of the norm, one proves continuity of $\underline{h}$ for all elements of $H^{\ast}(\underline{A}_{1}\odot\underline{A}_{2})$.

This proves that $\underline{h}$ can be extended to a continuous $\ast$-homomorphism $\underline{h}':H^{\ast}(\underline{A}_{1}\otimes\underline{A}_{2})\to\underline{B}$. Note that $\underline{h}'$ is unique, as $\underline{h}$ is unique by construction.

To complete the proof of the theorem, consider morphisms in $\mathbf{ucCSp}$
\begin{equation*}
(f_{i},\underline{f}_{i}): (X,\underline{B})\to(\mathcal{C}_{i},\underline{A}_{i}).
\end{equation*}
There exists a unique $h:X\to\mathcal{C}_{1}\times\mathcal{C}_{2}$, such that $f_{i}=\pi_{i}\circ h$. In particular $F^{\ast}_{i}=H^{\ast}P_{i}^{\ast}$. We are given continuous $\ast$-homomorphisms $\underline{f}_{i}:H^{\ast}P^{\ast}_{i}\underline{A}_{i}\to\underline{B}$. By the previous reasoning, there exists a unique $\underline{h}:H^{\ast}(\underline{A}_{1}\otimes\underline{A}_{2})\to\underline{B}$ such that $\underline{h}\circ H^{\ast}\underline{p}_{i}=\underline{f}_{i}$. The pair $(h,\underline{h})$ is the unique arrow such that $(f_{i},\underline{f}_{i})=(\pi_{i},\underline{p}_{i})\circ(h,\underline{h})$.
\end{enumerate}
\end{proof}

Using the product of the previous theorem, we can write down the descent morphism.

\begin{poe} \label{poe: descent}
In $\mathbf{ucCSp}$, the descent morphism is given by
\begin{equation*}
(h,\underline{h}):([\mathcal{C}_{1\vee2},\mathbf{Set}],\underline{A}_{1\vee2})\to([\mathcal{C}_{1}\times\mathcal{C}_{2},\mathbf{Set}],\underline{A}_{1}\otimes\underline{A}_{2}),
\end{equation*}
\begin{equation*}
h:\mathcal{C}_{1\vee2}\to\mathcal{C}_{1}\times\mathcal{C}_{2}\ \ C\mapsto(C\cap A_{1},C\cap A_{2}),
\end{equation*}
\begin{equation*}
\underline{h}:H^{\ast}(\underline{A}_{1}\otimes\underline{A}_{2})\to\underline{A}_{1\vee2},
\end{equation*}
\begin{equation*}
\underline{h}_{C}:(C\cap A_{1})\otimes(C\cap A_{2})\to C\ \ a\otimes b\mapsto a\cdot b.
\end{equation*}
\end{poe}

Using this description of the descent morphism, we now come to our main and final result, which relates the C*-algebraic version of Nuiten's sheaf condition to C*-independence.

\begin{tut} \label{tut: crux}
Let $(A,B)$ satisfy extended locality. Then the pair $(A,B)$ satisfies the sheaf condition in $\mathbf{ucCSp}$ iff the pair is C*-independent, and satisfies
\begin{equation} \label{equ: couni}
\forall C\in\mathcal{C}_{A\vee B}\ \ (C\cap A)\vee(C\cap B)=C.
\end{equation}
\end{tut}

\begin{proof}
Let $(A,B)$ satisfy the conditions of Theorem~\ref{tut: crux}. We will define an inverse $(j,\underline{j})$ for the descent morphism $(h,\underline{h})$. C*-independence implies strong locality. Combined with (\ref{equ: couni}), we deduce that the poset morphism $h$ from Proposition~\ref{poe: descent} is an isomorphism. The inverse for $h$ is given by $\vee$, so we simply define $j=\vee$. The poset map $j$ defines an essential geometric morphism $J$, with inverse image functor $J^{\ast}:[\mathcal{C}_{1\vee2},\mathbf{Set}]\to[\mathcal{C}_{1}\times\mathcal{C}_{2},\mathbf{Set}]$. In order to define
\begin{equation*}
\underline{j}:J^{\ast}\underline{A\vee B}\to\underline{A}\otimes\underline{B},
\end{equation*}
for each $(C,D)\in\mathcal{C}_{A}\times\mathcal{C}_{B}$ define
\begin{equation*}
\underline{j}_{(C,D)}:C\vee D\to C\otimes D
\end{equation*}
to be the inverse of the isomorphism of commutative C*-algebras $\underline{h}_{C\vee D}$. Here we used C*-independence and (\ref{equ: couni}) to deduce that $\underline{h}_{C\vee D}$ is indeed an isomorphism $C\otimes D\to C\vee D$. We conclude that the pair $(A,B)$ satisfies the sheaf condition.

The converse claim, to the effect that the sheaf condition implies the conditions of the theorem is straightforward to prove.
\end{proof}

\section{Discussion}

Because of the counit law (\ref{equ: couni}), the full sheaf condition cannot be assumed for physically reasonable nets $O\mapsto A(O)$. The decent morphism $(H,\underline{h})$ does satisfy the weaker conditions that $H$ is a local geometric morphism, and for each context of the form $C\vee D$, the $\ast$-homomorphism $\underline{h}_{C\vee D}$ is an isomorphism of C*-algebras.

The calculations in this paper used the ad hoc simplification where we considered only topoi of the form $Sh(X)$, where $X$ is a topological space. This is because the more complex product of topoi (to be precise, we can consider the product of locales) is then replaced by the simpler product of topological spaces. It would be of interest to investigate the sheaf condition using the larger category of all topoi with internal C*-algebras. In particular, it would be interesting how to see how this affects the counit law (\ref{equ: couni}).

Thinking in another direction it would also be of interest to connect the work in this paper to~\cite{bfir}, which takes a functorial perspective on C*-independence as well, but investigates it as a monoidal structure rather than as a sheaf condition.

\begin{appendix}

\section{Internal C*-algebras} \label{sec: internal}

Let $\mathcal{E}$ be a topos with natural number object, and $\uA$ an object of this topos. In addition, let $\qui$ denote the complexified rational numbers of $\mathcal{E}$. In the definition of a C*-algebra in a topos we make use of $\qui$ for scalar multiplication instead of $\underline{\mathbb{C}}$, the complexified Dedekind real numbers. This is because $\qui$ is preserved under the action of inverse image functors, whereas $\underline{\mathbb{C}}$ is generally not. 

In what follow we will use shorthand notation such as $\forall a,b\in\uA$ for $\forall a\in\uA,\ \forall b\in\uA$. We can now start with the definition of a C*-algebra in $\mathcal{E}$, based on~\cite{banmul3}. First of all, $\uA$ is a $\qui$-vector space. This means that there are arrows
\begin{equation*}
+:\uA\times\uA\to\uA,\ \ \cdot:\qui\times\uA\to\uA,\ \ 0:\underline{1}\to\uA,
\end{equation*}
defining addition, scalar multiplication and the constant $0$. These should satisfy the usual axioms for a vector space such as
\begin{equation*}
\forall a,b,c\in\uA\ \left((a+b)+c=a+(b+c)\right),
\end{equation*}
\begin{equation*}
\forall a\in\uA\ a+0=a.
\end{equation*}
In addition, there is a multiplication $\cdot:\uA\times\uA\to\uA$ satisfying the axioms which make $\uA$ into a $\qui$-algebra. We used the notation $\cdot$ for multiplication as well as scalar multiplication, hoping that this will not lead to confusion.
 
There is an arrow $\ast:\uA\to\uA$, which is involutive
\begin{equation*}
\forall a\in\uA\ (a^{\ast})^{\ast}=a,
\end{equation*}
conjugate linear,
\begin{equation*}
\forall a,b\in\uA\ (a+b)^{\ast}=a^{\ast}+ b^{\ast},
\end{equation*}
\begin{equation*}
\forall a\in\uA,\ \forall x\in\qui\ (x\cdot a)^{\ast}=\bar{x}\cdot a^{\ast},
\end{equation*}
where $\bar{(\cdot)}:\qui\to\qui$ is the conjugation map $x+iy\mapsto x-iy$. The involution is antimultiplicative
\begin{equation*}
\forall a,b\in\uA\ (a\cdot b)^{\ast}=b^{\ast}\cdot a^{\ast}.
\end{equation*}
In the topos $\mathbf{Set}$, the norm is defined as a map $\Vert\cdot\Vert: A\to[0,\infty)$. Equivalently, it can be described as a subset $N\subset A\times\mathbb{Q}^{+}$, where $(a,p)\in N$ iff $\Vert a\Vert<p$. For C*-algebras in arbitrary topoi, we use the subset description as it is formulated using rational numbers. A norm on $\uA$ is a subobject $\uN\subseteq\uA\times\qup$ satisfying the axioms stated below. The axiom
\begin{equation*}
\forall p\in\qup\ \ (0,p)\in\uN
\end{equation*}
expresses $\Vert0\Vert=0$. The fact that $\Vert a\Vert=0$ implies $a=0$, stating that a given semi-norm is in fact a norm, is expressed as
\begin{equation*}
\forall a\in\uA\left(\left(\forall p\in\qup\ (a,p)\in\uN\right)\to(a=0)\right).
\end{equation*}
Note that because of the second universal quantifier, this axiom does not fit within the constraints of geometric logic. The following two axioms express that the norm $\uN$ can be seen as a mapping $\Vert\cdot\Vert:\uA\to\underline{[0,\infty]}_{u}$ (see e.g.~\cite{vic05}). The subscript $u$ indicates that we are using upper real numbers here, one of the various kinds of real numbers in a topos. As the internal mathematics of a topos is constructive, different ways of constructing real numbers out of the rational numbers can result in different objects~\cite[Section D4.7]{jh1}. The first axiom excludes the possibility that the $\Vert a\Vert$ is equal to the upper real number $\infty$.
\begin{equation*}
\forall a\in\uA\ \exists p\in\qup\ \ (a,p)\in\uN,
\end{equation*}
\begin{equation*}
\forall a\in\uA\ \forall p\in\qup\ \left((a,p)\in\uN\leftrightarrow\left(\exists q\in\qup\ (p>q)\wedge((a,q)\in\uN)\right)\right).
\end{equation*}
The equality $\Vert a\Vert=\Vert a^{\ast}\Vert$ follows from the involutive property of $\ast$ and the axiom
\begin{equation*}
\forall a\in\uA\ \forall p\in\qup\ \left((a,p)\in\uN\to(a^{\ast},p)\in\uN\right).
\end{equation*}
The triangle inequality $\Vert a+b\Vert\leq\Vert a\Vert+\Vert b\Vert$ is expressed by the axiom
\begin{equation*}
\forall a,b\in\uA\ \forall p,q\in\qup\ \left(\left((a,p)\in\uN\wedge(b,q)\in\uN\right)\to(a+b,p+q)\in\uN\right).
\end{equation*}
Submultiplicativity of the norm $\Vert a\cdot b\Vert\leq\Vert a\Vert\cdot\Vert b\Vert$ is expressed by the axiom
\begin{equation*}
\forall a,b\in\uA\ \forall p,q\in\qup\ \left(\left((a,p)\in\uN\wedge(b,q)\in\uN\right)\to(a\cdot b,p\cdot q)\in\uN\right).
\end{equation*}
The property $\Vert x\cdot a\Vert=\vert x\vert\cdot\Vert a\Vert$ is expressed as
\begin{equation*}
\forall a\in\uA\ \forall x\in\qui\ \forall p,q\in\qup\ \left(\left((a,p)\in\uN\wedge(\vert x\vert<q)\right)\to(x\cdot a,p\cdot q)\in\uN\right),
\end{equation*}
where we used the modulus map
\begin{equation*}
\vert\cdot\vert:\qui\to\qup\ \ x+iy\mapsto x^{2}+y^{2}.
\end{equation*}
The special property $\Vert a\Vert^{2}=\Vert a\cdot a^{\ast}\Vert$ is given by
\begin{equation*}
\forall a\in\uA\ \forall p\in\qup\ \left((a,p)\in\uN\to(a\cdot a^{\ast},p^{2})\in\uN\right).
\end{equation*}
The algebra $\uA$ is required to be complete with respect to the norm $\uN$. This can be expressed using \textbf{Cauchy approximations}. Let $\mathcal{P}\uA$ denote the power object of $\uA$. A sequence $C:\underline{\mathbb{N}}\to\mathcal{P}\uA$ is a Cauchy approximation if it satisfies the following two axioms
\begin{equation*}
\forall n\in\underline{\mathbb{N}}\ \exists a\in\uA\ (a\in C(n)),
\end{equation*}
\begin{equation*}
\forall k\in\underline{\mathbb{N}}\ \exists m\in\underline{\mathbb{N}}\ \forall n,n'\geq m\ \left( (a\in C(n))\wedge(b\in C(n'))\to(a-b,1/k)\in\uN\right).
\end{equation*}
Note that the first axiom simply states that each set $C(n)$ is non-empty, whereas the second axiom is the characterising property of Cauchy sequences. The difference between Cauchy sequences and Cauchy approximations is that the first uses singleton subsets of the algebras, whereas the second uses non-empty sets. Note that for the second axiom we used the shorthand notation $\forall n,n'\geq m$, meaning
\begin{equation*}
\forall n\in\underline{\mathbb{N}}\ \forall n'\in\underline{\mathbb{N}}\ (n\geq m)\wedge(n'\geq m)\to.
\end{equation*}
The normed algebra $\uA$ is complete if each Cauchy approximation converges to a unique element of $\uA$. Given a Cauchy approximation $C$ and an element $a\in\uA$, we say that $C$ converges to $a$ iff
\begin{equation*}
\forall k\in\underline{\mathbb{N}}\ \exists m\in\underline{\mathbb{N}}\ \forall n\geq m\ \left(b\in C(n)\to(b-a,1/k)\in\uN\right).
\end{equation*}
We briefly use the following notation to reduce the size of the formulae involved. Given a sequence $C:\underline{\mathbb{N}}\to\mathcal{P}\uA$, let $\psi(C)$ denote the proposition that $C$ is a Cauchy approximation (i.e. the conjunction of the two axioms given above). For a sequence $C$ of subsets of $\uA$, and $a\in\uA$ let $\phi(C,a)$ denote the proposition stating that $C$ converges to $a$. The normed algebra $\uA$ is complete iff it satisfies
\begin{equation*}
\forall C\in\mathcal{P}\uA^{\underline{\mathbb{N}}}\ \psi(C)\to\left(\exists a\in\uA\ \phi(C,a)\right),
\end{equation*}
\begin{equation*}
\forall a,b\in\uA\ \forall C\in\mathcal{P}\uA^{\underline{\mathbb{N}}}\ \psi(C)\to\left(\phi(C,a)\wedge\phi(C,b)\to a=b\right).
\end{equation*}
This completes the definition of a C*-algebra in $\mathcal{E}$.\\

A C*-algebra is called commutative if it satisfies the additional axiom
\begin{equation*}
\forall a,b\in\uA\ a\cdot b=b\cdot a.
\end{equation*}
A C*-algebra is called unital, if there is a constant $1:\underline{1}\to\uA$, satisfying the axioms
\begin{equation*}
\forall a\in\uA\ \ a\cdot1=a=1\cdot a,
\end{equation*}
\begin{equation*}
\forall p\in\qup\ \left((p>1)\to(1,p)\in\uN\right).
\end{equation*}

\end{appendix}

\end{document}